\documentclass[a4paper,UKenglish,cleveref, autoref]{lipics-v2019}

\nolinenumbers
\hideLIPIcs

\usepackage[utf8]{inputenc}

\usepackage{newunicodechar}
\newunicodechar{α}{\ensuremath{\alpha}\xspace}
\newunicodechar{β}{\ensuremath{\beta}\xspace}
\newunicodechar{γ}{\ensuremath{\gamma}\xspace}
\newunicodechar{σ}{\ensuremath{\sigma}\xspace}
\newunicodechar{δ}{\ensuremath{\delta}\xspace}
\newunicodechar{τ}{\ensuremath{\tau}\xspace}
\newunicodechar{φ}{\ensuremath{\varphi}\xspace}
\newunicodechar{ε}{\ensuremath{\varepsilon}\xspace}
\newunicodechar{π}{\ensuremath{\pi}\xspace}
\newunicodechar{ψ}{\ensuremath{\psi}\xspace}
\newunicodechar{θ}{\ensuremath{\theta}\xspace}
\newunicodechar{λ}{\ensuremath{\lambda}\xspace}
\newunicodechar{ξ}{\ensuremath{\xi}\xspace}

\newunicodechar{Λ}{\ensuremath{\Lambda}\xspace}
\newunicodechar{Ψ}{\ensuremath{\Psi}\xspace}
\newunicodechar{Γ}{\ensuremath{\Gamma}\xspace}
\newunicodechar{Φ}{\ensuremath{\Phi}\xspace}
\newunicodechar{Π}{\ensuremath{\Pi}\xspace}
\newunicodechar{Δ}{\ensuremath{\Delta}\xspace}
\newunicodechar{Θ}{\ensuremath{\Theta}\xspace}

\newunicodechar{∃}{\ensuremath{\exists}\xspace}
\newunicodechar{∀}{\ensuremath{\forall}\xspace}
\newunicodechar{ℕ}{\ensuremath{\mathbb{N}}\xspace}
\newunicodechar{ℤ}{\ensuremath{\mathbb{Z}}\xspace}
\newunicodechar{⇐}{\ensuremath{\Leftarrow}\xspace}
\newunicodechar{⇔}{\ensuremath{\Leftrightarrow}\xspace}
\newunicodechar{⇒}{\ensuremath{\Rightarrow}\xspace}
\newunicodechar{↦}{\ensuremath{\mapsto}\xspace}

\newunicodechar{∨}{\ensuremath{\lor}\xspace}
\newunicodechar{∧}{\ensuremath{\land}\xspace}
\newunicodechar{⊥}{\ensuremath{\bot}\xspace}

\usepackage{color}

\usepackage{xspace}
\usepackage{amsmath}
\usepackage{amssymb}
\usepackage{multicol}
\usepackage{braket}
\usepackage{mathtools}
\usepackage{enumerate}
\usepackage{tikz}
\usetikzlibrary{chains}
\usetikzlibrary{calc}
\usepackage{booktabs}
\usepackage{amsthm}
\usepackage{url}
\usepackage{hyperref}

\newenvironment{proofsketch}{%
  \renewcommand{\proofname}{Proof Sketch}\proof}{\endproof}

\makeatletter
\newcommand*{\defeq}{\mathrel{%
  \rlap{\raisebox{0.3ex}{$\m@th\cdot$}}%
  \raisebox{-0.3ex}{$\m@th\cdot$}}%
  =}\makeatother
\makeatletter
\newcommand*{\eqdef}{=\mathrel{%
  \raisebox{0.3ex}{$\m@th\cdot$}%
  \llap{\raisebox{-0.3ex}{$\m@th\cdot$}}}%
  }\makeatother
\makeatletter
\newcommand*{\defeqv}{\mathrel{%
  \rlap{\raisebox{0.3ex}{$\m@th\cdot$}}%
  \raisebox{-0.3ex}{$\m@th\cdot$}}%
  \Longleftrightarrow}\makeatother
\makeatletter
\newcommand*{\eqvdef}{\Longleftrightarrow\mathrel{%
  \raisebox{0.3ex}{$\m@th\cdot$}%
  \llap{\raisebox{-0.3ex}{$\m@th\cdot$}}}%
  }\makeatother

\makeatletter
\newcommand*{\defdefeq}{\mathrel{%
  \rlap{\raisebox{0.3ex}{$\m@th\cdot\cdot$}}%
  \raisebox{-0.3ex}{$\m@th\cdot\cdot$}}%
  =}\makeatother

\newcommand{\prob}[1]{\classFont{#1}}
\newcommand{\perm}{\protect\ensuremath{\prob{PERMANENT}}\xspace}
\newcommand{\cycc}{\protect\ensuremath{\prob{CYCLE-COVER}}\xspace}
\newcommand{\perfmat}{\protect\ensuremath{\prob{PERFECT-MATCHING}}\xspace}
\newcommand{\impmat}{\protect\ensuremath{\prob{IMPERFECT-MATCHING}}\xspace}

\newcommand{\ssat}{\protect\ensuremath{\Sigma_1\classFont{SAT}}\xspace}
\newcommand{\cnf}{\protect\ensuremath{\classFont{CNF}}\xspace}
\newcommand{\tcnf}{\protect\ensuremath{\classFont{3CNF}}\xspace}
\newcommand{\scnf}{\protect\ensuremath{\Sigma_1\classFont{CNF}}\xspace}
\newcommand{\stcnf}{\protect\ensuremath{\Sigma_1\classFont{3CNF}}\xspace}
\newcommand{\bcnf}{\protect\ensuremath{\classFont{2CNF}}\xspace}

\newcommand{\dep}[1]{\protect\ensuremath{=\!\!(#1)}\xspace}
\newcommand{\cTeam}[1]{\protect\ensuremath{\#{#1}^{\textnormal{team}}}\xspace}

\newcommand\classFont[1]{\textnormal{#1}}
\renewcommand{\P}{\protect\ensuremath{\classFont{P}}\xspace}
\newcommand{\NP}{\protect\ensuremath{\classFont{NP}}\xspace}
\newcommand{\cP}{\protect\ensuremath{{\classFont{\#P}}}\xspace}
\newcommand{\totP}{\protect\ensuremath{{\classFont{TotP}}}\xspace}
\newcommand{\fP}{\protect\ensuremath{{\classFont{FP}}}\xspace}
\newcommand{\PP}{\protect\ensuremath{\classFont{PP}}\xspace}
\newcommand{\coNP}{\protect\ensuremath{\classFont{coNP}}\xspace}

\newcommand{\cAC}{\protect\ensuremath{{\#\classFont{AC}^0}}\xspace}

\newcommand{\FO}{\protect\ensuremath{\classFont{FO}}\xspace}

\newcommand{\FTC}{\protect\ensuremath{\classFont{FTC}^0}\xspace}

\newcommand{\struc}{\protect\ensuremath{\textnormal{STRUC}}}
\newcommand{\tString}{\protect\ensuremath{τ_{\textnormal{string}}}\xspace}

\newcommand{\tu}[1]{\overline{#1}}

\newcommand\restr[2]{{
  \ensuremath{\left.\kern-\nulldelimiterspace 
  #1 
  \vphantom{\big|} 
  \right|_{#2} 
  }}}

\newcommand{\ie}{i.e.\@\xspace}
\newcommand{\eg}{e.g.\@\xspace}

\newcommand{\wrt}{with respect to\@\xspace}

\newcommand{\stfa}{such that for all\@\xspace}

\newcommand{\ifff}{if and only if\@\xspace}

\newcommand{\cFO}{\protect\ensuremath{\#\FO}\xspace}

\newcommand{\rel}{\mathrm{rel}}
\newcommand{\dom}{\mathrm{dom}}
\newcommand{\func}{\mathrm{func}}
\newcommand{\cSigmaiRel}[1]{\protect\ensuremath{{\#\Sigma_{#1}^{\rel}}}\xspace}
\newcommand{\cPiiRel}[1]{\protect\ensuremath{{\#\Pi_{#1}^{\rel}}}\xspace}

\newcommand{\Sigmai}[1]{{\Sigma_{#1}}}

\newcommand{\calA}{\mathcal{A}}

\newcommand{\calC}{\mathcal{C}}

\newcommand{\enc}{\mathrm{enc}}

\newcommand{\problemdef}[3]{%
\begin{center}
\begin{tabular}{rp{6cm}}\toprule
 \textbf{Problem}:& #1\\
 \textbf{Input}:& #2\\
 \textbf{Output}:& #3\\
 \bottomrule
\end{tabular}
\end{center}
}

\usepackage{microtype}

\title{Counting of Teams in First-Order Team Logics}
\titlerunning{Counting of Teams in First-Order Team Logics}

\author{Anselm Haak}{Theoretische Informatik, Leibniz Universität Hannover, Appelstraße, D-30167, Germany}{haak@thi.uni-hannover.de}{https://orcid.org/0000-0003-1031-5922}{}
\author{Juha Kontinen}{Department of Mathematics and Statistics, University of Helsinki, Pietari Kalmin katu 5, 00014, Finland}{juha.kontinen@helsinki.fi}{https://orcid.org/0000-0003-0115-5154}{}
\author{Fabian Müller}{Theoretische Informatik, Leibniz Universität Hannover, Appelstraße, D-30167, Germany}{fabian.mueller@thi.uni-hannover.de}{www.thi.uni-hannover.de}{}
\author{Heribert Vollmer}{Theoretische Informatik, Leibniz Universität Hannover, Appelstraße, D-30167, Germany}{vollmer@thi.uni-hannover.de}{https://orcid.org/0000-0002-9292-1960}{}
\author{Fan Yang}{Department of Mathematics and Statistics, University of Helsinki, Pietari Kalmin katu 5, 00014, Finland}{fan.yang@helsinki.fi}{https://orcid.org/0000-0003-0392-6522}{}

\authorrunning{A. Haak, J. Kontinen, F. Müller, H. Vollmer, F. Yang}

\Copyright{Anselm Haak, Juha Kontinen, Fabian Müller, Heribert Vollmer, Fan Yang}



\relatedversion{A full version of the paper is available at \url{https://arxiv.org/abs/1902.00246}.}

\funding{All authors were supported by DAAD grant 57348395.\\
\emph{Anselm Haak, Fabian Müller, Heribert Vollmer}: Supported by DFG VO 630/8-1.\\
\emph{Juha Kontinen, Fan Yang}: Supported by grant 308712 of the Academy of Finland.}

\acknowledgements{We thank the anonymous referees for helpful comments.}

\EventEditors{Peter Rossmanith, Pinar Heggernes, and Joost-Pieter Katoen}
\EventNoEds{3}
\EventLongTitle{44th International Symposium on Mathematical Foundations of Computer Science (MFCS 2019)}
\EventShortTitle{MFCS 2019}
\EventAcronym{MFCS}
\EventYear{2019}
\EventDate{August 26--30, 2019}
\EventLocation{Aachen, Germany}
\EventLogo{}
\SeriesVolume{138}
\ArticleNo{42}

\begin{document}

\maketitle

\begin{abstract}
We study descriptive complexity of counting complexity classes in the range from \cP to $\#\cdot \NP$. A corollary of Fagin's characterization of NP by existential second-order logic is that  \cP can be logically described as the class of functions counting satisfying assignments to free relation variables in first-order formulae. In this paper we extend this study to classes beyond \cP and extensions of first-order logic with team semantics.  
These team-based logics are closely related to existential second-order logic and its fragments, hence our results also shed light on the complexity of counting for extensions of first-order logic in Tarski's semantics.
Our results show that the class $\#\cdot \NP$ can be logically characterized by independence logic and existential second-order logic, whereas dependence logic and inclusion logic give rise to 
subclasses of $\#\cdot \NP$ and \cP, respectively. We also study the function class generated by inclusion logic and relate it to the complexity class $\totP\subseteq \cP$. Our main technical result shows that the problem of counting satisfying assignments for monotone Boolean $\Sigma_1$-formulae is $\#\cdot \NP$-complete with respect to Turing reductions as well as complete for the function class generated by dependence logic with respect to first-order reductions. 

\end{abstract}

\newpage

\section{Introduction}

The question of the power of counting arises in propositional and predicate logic in a number of contexts. Counting the number of satisfying assignments for a given propositional formula, $\#\mathrm{SAT}$, is complete for Valiant's class \cP of functions counting accepting paths of nondeterministic polynomial-time Turing machines \cite{DBLP:journals/tcs/Valiant79}. Valiant also proved that $\#\mathrm{SAT}$ even remains complete when restricted to monotone \bcnf-formulae.

The class $\cP$ can be seen as the counting analogue of $\NP$, which was shown by Fagin \cite{fagingeneralized} to correspond to existential second order logic, where the quantified relation encodes accepting computation paths of NP-machines. Hence, if we define $\cFO^\rel$ to count satisfying assignments to free relational variables in first-order (\FO-) formulae, we obtain $\cFO^\rel=\cP$. This result has been refined to prefix classes of \FO showing, e.g., that $\cPiiRel{2} = \cP$ \cite{DBLP:journals/jcss/SalujaST95}. 

If we define $\cFO^\func$ in the same fashion as $\cFO^\rel$ except that we count assignments to function variables instead of relation variables, then obviously $\cFO^\func=\cP$. The situation changes for the prefix classes, though. In particular, unlike for $\#\Pi_1^\rel$ , it holds that $\#\Pi_1^\func = \cP$, and, remarkably, also arithmetic circuit classes like $\cAC$ can be characterized in this context \cite{DBLP:conf/csl/0001HKV16}. 

In this paper we consider a different model-theoretic approach to the study of counting processes using so-called team-based logics. In these logics, formulae with free variables are evaluated not for single assignments to these variables but for \emph{sets} of such assignments (called \emph{teams}). Logics with team semantics have been developed for the study of various dependence and independence concepts important in many areas such as (probabilistic) databases and Bayesian networks (see, e.g. \cite{DBLP:journals/iandc/HannulaK16,DBLP:conf/wollic/CoranderHKPV16,JELIA}) for which model counting is an important inference task (see, e.g., \cite{DBLP:journals/corr/BacchusDP14,DBLP:conf/lics/KuusistoL18}). In addition, team-based logics have interesting connections to a wide range of areas such as formal semantics of natural language \cite{Ciardelli2015}, social choice theory \cite{PacuitYang2016}, and quantum information theory \cite{PV15}.

In team semantics, a first-order formula is satisfied by a team \ifff all its members satisfy the formula individually. Interest in teams stems from the introduction of different logical atoms describing properties of teams, called team atoms, such as the value of a variable being functionally dependent on other variables (characterized by the \emph{dependence} atom $\dep{\bar{x},y}$), a variable being independent from other variables (characterized by the \emph{independence} atom $\bar{y}\bot_{\bar{x}}\bar{z}$), and the values of a variable occurring as values of some other variable (characterized by the \emph{inclusion} atom $\bar{x}\subseteq\bar{y}$), etc. (\cite{DBLP:books/daglib/0030191,DBLP:journals/sLogica/GradelV13,DBLP:journals/apal/Galliani12}). 

We initiate in this paper the study of counting for team-based logics. In our proofs we utilize the known correspondences between team-based logics and existential second-order logic ($\Sigma^1_1$) and its fragments (see Theorem \ref{expr}). We want to stress that our results are also novel for existential second-order logic and its fragments, and that there is no natural way to carry out the study of the function class generated by inclusion logic, that is, \FO extended by the inclusion atom, in Tarskian semantics.

We define $\cTeam{\FO}$ to be the class of functions counting teams that satisfy a given FO-formula, and similarly for extensions of FO by team atoms. 
Making use of different team atoms, we give a characterization of $\#\cdot\NP$. 
While it is relatively easy to see that with every finite set $A$ of \NP-definable team atoms, the class $\cTeam{\FO(A)}$ stays a subclass of $\#\cdot\NP$ (Toda's generalization of $\cP$, see \cite{DBLP:journals/sigact/HemaspaandraV95} for a survey of counting classes like these), we show that \FO extended with the independence atom is actually sufficient to characterize the full class $\#\cdot\NP$:
\[\cTeam{\FO(\bot)}=\#\Sigma_1^1=\#\cdot\NP.\]

The situation with inclusion logic and dependence logic is more complex due to their strong closure properties: satisfaction of formulae is closed under union for inclusion logic and is closed downwards for dependence logic. We show that $\cTeam{\FO(\subseteq)}$ is a subclass of \totP, which is a class of counting problems with easy decision versions. Note that \totP is a strict subclass of \cP unless $\P = \NP$ and consequently the same holds for $\cTeam{\FO(\subseteq)}$. Furthermore, $\cTeam{\FO(\dep{\dots})}$ is a subclass of $\#\cdot \NP$, which we believe to be strict as well.
Interestingly, both classes contain complete problems from their respective superclasses. In establishing this result for dependence logic, we introduce an interesting class of monotone quantified Boolean formulae and show that the corresponding counting problem where the all-0-assignment is not counted, $\#\scnf_*^-$, is $\#\cdot\NP$-complete. In order to prove $\#\cdot\NP$-completeness we also show that the more natural problem of counting all satisfying assignments of the same class of formulae is $\#\cdot \NP$-complete by introducing a new technique of simultaneous reductions between pairs of counting problems, which we hope will also be useful in other contexts.

For inclusion logic we show that the well-known $\cP$-complete problem $\#\bcnf^+ $ is in $\cTeam{\FO(\subseteq)}$ and that the problem of counting assignments for existentially quantified dual-Horn formulae is hard for $\cTeam{\FO(\subseteq)}$.

In related previous work, so-called weighted logics have been used to logically characterize counting complexity classes \cite{DBLP:conf/lics/ArenasMR17}, and the decision-problem analogue \PP of $\cP$ and the counting hierarchy have been logically characterized in \cite{DBLP:journals/tocl/Kontinen09,DBLP:journals/iandc/KontinenN11,DurandEKV15}. Counting classes from circuit complexity beyond $\cAC$ have been logically characterized in \cite{DBLP:conf/lics/0001HV18}.

Due to space restrictions, we only give proof sketches for some theorems, and detailed proofs for all results throughout the paper are deferred to the full version of this paper.

\section{Definitions and Preliminaries}\label{sec:prelims}

\paragraph*{First-order Logic and Team Semantics}

Let us start by recalling the syntax of first-order logic (\FO). In this work, we only consider relational vocabularies (i.e., vocabularies with no function or constant symbols), and thus the only first-order terms are variables. 
Formulae of first-order logic are defined by the following grammar:
\[φ \defdefeq φ ∧ φ \mid φ ∨ φ \mid ∃x φ \mid ∀x φ \mid R(\tu x) \mid\! \neg R(\tu x) \mid x = y \mid x \neq y \label{eq:snytax}\tag{$\star$}\]
where $x,y$ are variables, $R$ is a relation symbol, and $\tu x$ is a tuple of the appropriate number of variables.

The set of {\em free variables} of a formula $\varphi$ is defined as usual, and we sometimes write $\varphi(x_1,\dots,x_k)$ to emphasize that the free variables of $\varphi$ are among $x_1,\dots,x_k$. A formula with no free variable is called a {\em sentence}. For any $k$, the fragment $\Sigma_k$ of \FO consists of all formulae of the form $\exists \tu{x_1}\forall \tu{x_2}\dots Q\tu{x_k}\varphi$, where $\varphi$ is quantifier-free and $Q$ is either $\exists$ or $\forall$ depending on whether $k$ is odd or even; similarly, the fragment $\Pi_k$ is defined as the class of all formulae $\forall \tu{x_1}\exists\tu{x_2}\dots Q\tu{x_k}\varphi$ in prenex normal form with a quantifier prefix with $k$ alternations starting with universal quantifiers.

We only consider finite structures with a finite relational vocabulary $\sigma$. 
 Denote the class of all such structures by $\struc[\sigma]$, and let $\textrm{dom}(\calA)$ denote the universe of a $\sigma$-structure $\calA$. We will always use structures with universe $\{0, 1, \dots, n-1\}$ for some $n \in \mathbb{N}\setminus\{0\}$. We assume that our structures contain a \emph{built-in binary relation} $\leq$ and \emph{ternary relations} $+, \times$ with the usual interpretation, i.e., $\leq$ is interpreted in a model of any size as the ``less than or equal to'' relation on $\mathbb{N}$, $+$ is interpreted as addition and $\times$ as multiplication. 
 We write $\enc_\sigma(\mathcal{A})$ for the standard binary encoding of a $\sigma$-structure $\mathcal{A}$ (see \eg, \cite{DBLP:books/daglib/0095988}): Relations are encoded row by row by listing their truth values as $0$'s and $1$'s.
 The whole structure is encoded by the concatenation of the encodings of its relations.

We assume that the reader is familiar with the usual Tarskian semantics for first-order formulae, in which formulae are evaluated with respect to single assignments of a structure. In this paper, we also consider so-called {\em team semantics} for first-order formulae, in which formulae are evaluated with respect to teams. A {\em team} is a {\em set} of assignments of a structure, that is, a set of functions $s \colon \{x_1, \dots, x_k \,\} \to \textrm{dom}(\calA)$, where we call $\{x_1, \dots, x_k \,\}$ the domain of the team. Note that the empty set $\emptyset$ is a team, called \emph{empty team}, and the singleton $\{\emptyset\}$ containing only the empty assignment is also a team. We denote by $\textrm{team}(\calA, (x_1, \dots, x_k))$ the set of all teams over a structure $\calA$ with the domain $\{x_1, \dots, x_k\}$. Due to certain connections between team logics and second-order logic it is often helpful to view teams as relations. Let $\mathcal{A}$ be a structure and $X$ a team of $\mathcal{A}$ with domain $\{x_1, \dots, x_k\}$. $\mathcal{A}$ and $X$  induce the $k$-ary relation $\rel(X)$ on $\dom(\mathcal{A})$ defined as
\[\rel(X) \defeq \{(s(x_1), \dots, s(x_n)) \mid s \in X\}.\]
Furthermore, for any subset $V \subseteq \{x_1, \dots, x_k\}$ of variables we define
\[\restr{X}{V} \defeq \{\restr{s}{V} \mid s \in X\},\]
the restriction of team $X$ to domain $V$.

We define inductively the notion of a team $X$ with domain $\{x_1,\dots,x_k\}$ of a structure $\mathcal{A}$ with $A\defeq\textrm{dom}(\calA)$ satisfying an \FO-formula $\varphi(x_1,\dots,x_k)$, denoted by $\mathcal{A}\models_X\varphi$, as follows:
\begin{itemize}
  \item $\mathcal{A} \models_X α$ for $\alpha$ an atomic formula \ifff for all $s \in X$, $\mathcal{A} \models_s α$ in the usual Tarskian semantics sense.
  \item $\mathcal A \models_X \varphi \vee \psi$ \ifff there are teams $Y,Z\subseteq X$ such that $Y \cup Z = X$, $\mathcal A \models_Y \varphi$ and $\mathcal A \models_Z \psi$.
  \item $\mathcal A \models_X \varphi \wedge \psi$ \ifff $\mathcal A \models_X \varphi$ and $\mathcal A \models_X \psi$.
  \item $\mathcal A \models_X \exists x \varphi$ \ifff there exists a function $F \colon X \to \mathcal P(A)\setminus\{\emptyset\}$, called \emph{supplementing function}, such that $\mathcal A \models_{X[F/x]} \varphi$, where
  \begin{align*}
  	X[F/x]&=\{s[a/x]\mid s\in X \text{ and }a\in F(s)\}& \text{ and } ~
  	s[a/x](y)= 
  	\begin{cases}
  	a,& \text{ if } x= y, \\
  	s(y),& \text{ else}.
  	\end{cases}
  \end{align*}
  
  \item $\mathcal A \models_X \forall x \varphi$ \ifff $\mathcal A \models_{X[A/x]} \varphi$, where
  \(X[A/x]=\{s[a/x]\mid s\in X\text{ and }a\in A\}.\)
\end{itemize}

\noindent A sentence $\varphi$ is said to be \emph{true} in $\calA$, written $\calA\models\varphi$, if $\calA\models_{\{\emptyset\}}\varphi$. 

\FO-formulae $\varphi$ are flat over team semantics, i.e., $\mathcal{A} \models_X \varphi$, \ifff $\mathcal{A} \models_{s} \varphi$ for all $s \in X$. In this sense, team semantics is conservative over \FO-formulae. We now extend first-order logic by sets of atomic formulae which are not flat. For any sequence $\tu{x}$ of variables and variable $y$, the string $\dep{\tu{x}, y}$ is called a {\em dependence atom}. For any sequences $\tu{x}, \tu{y}, \tu{z}$ of variables, the string $\tu{y} \perp_{\tu{x}}\tu{z}$ is called an {\em independence atom}. For any two sequences $\tu{x}$ and $\tu{y}$ of variables of the same length, the string $\tu{x} \subseteq \tu{y}$ is called an {\em inclusion atom}.
The team semantics of these atoms is defined as follows:
\begin{itemize}
  \item $\mathcal A \models_X \dep{\tu{x},y},$ \ifff for all $s, s' \in X$, if $s(\tu{x})=s'(\tu{x})$, then $s(y)=s'(y)$.
  \item $\mathcal A \models_X \tu{y} \bot_{\tu{x}} \tu{z}$ \ifff for all $s, s' \in X$ such that $s(\tu{x})=s'(\tu{x})$, there exists $s''\in X$ such that $s''(\tu{x})=s(\tu{x})$, $s''(\tu{y})=s(\tu{y})$ and $s''(\tu{z})=s'(\tu{z})$.
  \item $\mathcal A \models_X \tu{x} \subseteq \tu{y}$ \ifff for all $s \in X$ there is $s' \in X$ such that $s(\tu{x})=s'(\tu{y})$.
\end{itemize}



For any subset $A \subseteq \{\dep{\dots}, ⊥, \subseteq\}$, we define $\FO(A)$ as first-order logic extended by the respective atoms, and refer to such a logic as {\em team-based logic}. More precisely we extend the grammar (\ref{eq:snytax}) by adding a rule for each atom in $A$. For example for $\FO(\{\subseteq\})$ we add the rule 
\[\varphi\defdefeq \tu{x} \subseteq \tu{y},\]
where $\tu{x},\tu{y}$ are tuples of variables. For convenience, we often omit the curly brackets and write for example $\FO(\subseteq)$ instead of $\FO(\{\subseteq\})$.

The team-based logic $\FO(\dep{\dots})$ is known in the literature as {\em dependence logic} \cite{DBLP:books/daglib/0030191}, $\FO(\perp)$ as {\em independence logic} \cite{DBLP:journals/sLogica/GradelV13} and $\FO(\subseteq)$ as {\em inclusion logic} \cite{DBLP:journals/apal/Galliani12}.
We recall some basic properties of these logics from \cite{DBLP:books/daglib/0030191,DBLP:journals/sLogica/GradelV13,DBLP:journals/apal/Galliani12}: Formulae of $\FO(\dep{\dots})$ are {\em closed downwards}, i.e., $\calA\models_X\varphi$ and $Y\subseteq X$ implies $\calA\models_Y\varphi$, formulae of $\FO(\subseteq)$ are {\em closed under unions}, i.e., $\calA\models_X\varphi$ and $\calA\models_Y\varphi$ implies $\calA\models_{X\cup Y}\varphi$, and formulae of any of these logics have the {\em empty team property}, i.e., $\calA\models_\emptyset\varphi$ always holds.

The above atoms expressing team properties can be generalized, as we will do next. Let us first recall below the definition of generalized quantifiers, where we follow the notations from \cite{DBLP:conf/mfcs/KontinenKV16,DBLP:journals/jolli/Kuusisto15}.
\begin{definition}
	Let $i_1, \dots, i_n$ ($n>0$) be a sequence of positive integers, and $\sigma$ a vocabulary consisting of an $i_j$-ary relation symbol for each $1\leq j\leq n$. A \emph{generalized quantifier} of type $(i_1, \dots, i_n)$ is a class $\mathcal{C}$ of $\sigma$-structures $(A,B_1, \dots, B_n)$ such that the following conditions hold:
	\begin{enumerate}
		\item $A\neq \emptyset$ and for each $1\leq j\leq n$, we have $B_j\subseteq A^{i_j}$.
		\item $\mathcal{C}$ is closed under isomorphisms, that is, if $(A',B_1', \dots, B_n')\in \mathcal{C}$ is isomorphic to $(A,B_1, \dots, B_n)$, then $(A', B_1',\dots, B_n')\in\mathcal{C}$.		
	\end{enumerate}
 
Let $Q$ be a generalized quantifier of type $(i_1,\dots,i_n)$. 
Let us extend the syntax of first-order logic with an expression $A_Q(\tu{x_1},\dots, \tu{x_n})$, where each $\tu{x_j}$ is a tuple of variables of length $i_j$ and $\textrm{Vars}(\tu{x_i})$ is the set of variables in $x_i$. We call $A_Q$ a \emph{generalized (dependency) atom} (of type $(i_1, \dots, i_n)$), and its team semantics is defined as:
\[
\mathcal{A} \models_X A_Q(\tu{x_1},\dots, \tu{x_n}) \textrm{ \ifff }
(\rel(\restr{X}{\textrm{Vars}(\tu{x_1})}), \dots, \rel(\restr{X}{\textrm{Vars}(\tu{x_n})})) \in Q^{\mathcal{A}},\]
where $Q^{\mathcal{A}}=\{(B_1, \dots, B_n)\mid (\textnormal{dom}(\mathcal{A}), B_1, \dots, B_n) \in Q \}$.
\end{definition}

We say that a generalized dependency atom $A_Q$ is $\NP$-definable if there is an $\NP$-algorithm that decides for a given structure $\mathcal{A}$ and a given team $X$ whether $\mathcal{A} \models_X A_Q(\tu{x_1},\dots, \tu{x_n})$ holds or not. A set $A$ of generalized atoms is $\NP$-definable if every $a \in A$ is $\NP$-definable. For example, the set $A=\{\dep{\dots}, ⊥, \subseteq\}$ is $\NP$-definable.

Many results in this paper are based on the expressive power of the logics defined above, that we shall now recall.
We first recall some notions and notations. Existential second-order logic ($\Sigma^1_1$) consists of formulae of the form
\(\exists R_1\dots\exists R_k\varphi,\) 
where $\varphi$ is an \FO-formula. Let $\sigma$ be a vocabulary. We write $\sigma(R)$ for the vocabulary that arises by adding a fresh relation symbol $R$ to σ, and we sometimes write $\varphi(R)$ to emphasize that the relation symbol $R$ occurs in the $\sigma(R)$-formula $\varphi$. If $\calA$ is a $\sigma$-structure, we write $(\mathcal{A},Q)$ for $\mathcal{A}$ expanded into a $\sigma(R)$-structure where the new $k$-relation symbol $R$ is interpreted as $Q\subseteq\textrm{dom}(\calA)^k$.
 A $\sigma(R)$-sentence $\varphi(R)$ of $\Sigma^1_1$ is said to be \emph{downward monotone} with respect to $R$ if 
$(\calA,Q)\models\varphi(R)\text{ and }Q'\subseteq Q$ imply $(\calA,Q')\models\varphi(R)$.
It is known that $\varphi(R)$ is downward monotone with respect to $R$ if and only if $\varphi(R)$ is equivalent to a sentence where $R$ occurs only negatively (see e.g., \cite{DBLP:journals/jolli/KontinenV09}).

\begin{theorem}[see \cite{DBLP:journals/apal/Galliani12,DBLP:journals/jolli/KontinenV09,DBLP:conf/csl/GallianiH13}]\label{ind2sigma11}\label{expr}
\ \\\vspace{-8pt}
\begin{enumerate}
\item\label{sig=ind}
 For every $\sigma$-formula $\varphi$ of  $\FO(\perp)$, there is an $\sigma(R)$-sentence $\psi(R)$ of $\Sigma^1_1$ such that for all $\sigma$-structures $\calA$ and teams $X$,
\begin{equation}\label{sig11_ind_d}
\calA\models_X\varphi\iff (\calA,\textnormal{rel}(X))\models\psi(R).
\end{equation}
Conversely, for every $\sigma(R)$-sentence $\psi(R)$ of $\Sigma^1_1$, there is a $\sigma$-formula $\varphi$ of $\FO(\perp)$ such that (\ref{sig11_ind_d}) holds for all $\sigma$-structures $\calA$ and {\em non-empty} teams $X$. 
\item\label{dep=esor-} The same as the above holds for formulae of $\FO(\dep{\dots})$ as well, except that in both directions for $\FO(\dep{\dots})$ the relation symbol $R$ is assumed to occur only negatively in the sentence $\psi(R)$.
\item In particular, over sentences both $\FO(\perp)$ and $\FO(\dep{\dots})$ are expressively equivalent to $\Sigma^1_1$, in the sense that every $\sigma$-sentence of $\FO(\perp)$ (or $\FO(\dep{\dots})$) is equivalent to a $\sigma$-sentence $\psi$ of $\Sigma^1_1$, i.e., for any $\sigma$-structure $\calA$,
\[\calA\models\varphi\iff\calA\models\psi,\]
and vice versa. As a consequence of Fagin's Theorem (see \cite{fagingeneralized}), over finite structures both $\FO(\perp)$ and $\FO(\dep{\dots})$ capture \NP.
\item\label{inc=lfp} For any $\sigma$-formula $\varphi(x_1,\dots,x_k)$ of $\FO(\subseteq)$, there exists a $\sigma(R)$-formula $\psi(R)$ of positive greatest fixed point logic ($\classFont{posGFP}$) such that for all $\sigma$-structures $\calA$ and teams $X$,
\[\calA \models_X \varphi\iff \calA,\textnormal{rel}(X) \models_s \psi(R)\text{ for all }s \in X;\]
and vice versa.
In particular, over sentences $\FO(\subseteq)$ is expressively equivalent to $\classFont{posGFP}$. As a consequence of Immerman's Theorem (see \cite{DBLP:journals/iandc/Immerman86}), over finite structures, $\FO(\subseteq)$ is expressively equivalent to least fixed point logic ($\classFont{LFP}$). Thus, by \cite{DBLP:journals/iandc/Immerman86,DBLP:conf/stoc/Vardi82}, over ordered finite structures, $\FO(\subseteq)$ captures $\P$.
\item\label{myopic} 
Let $\varphi(R)$ be a \emph{myopic} $\sigma$-formula, that is, $\varphi(R)=\forall \tu{x} (R(\tu{x})\rightarrow \psi(R,\tu{x}))$, where $\psi$ is a first order $\sigma$-formula with only positive occurrences of $R$. Then there exists a $\sigma$-formula $\chi \in \FO(\subseteq)$ such that for all $\sigma$-structures $\mathcal{A}$ and all teams $X$:
\[\mathcal{A} \models_X \chi(\tu{x}) \Leftrightarrow \mathcal{A}, \textnormal{rel}(X) \models \varphi(R). \]
 \end{enumerate}
\end{theorem}

\paragraph*{Propositional and Quantified Boolean formulae}

In this paper, we will also consider certain classes of propositional and quantified Boolean formulae. As usual, we use \cnf to denote the class of propositional formulae in conjunctive normal form and $k\classFont{-}\cnf$ to denote the class of propositional formulae in conjunctive normal form where each clause contains at most $k$ literals. A formula in $\cnf$ is in the class $\textnormal{DualHorn}$ if each of its clauses contains at most one negative literal.

For a class $\mathcal{C}$ of Boolean formulae, we denote by $\Sigma_1\mathcal{C}$ the class of quantified Boolean formulae in prenex normal form with only existential quantifiers where the quantifier-free part is an element of $\mathcal{C}$.

For a class $\mathcal{C}$ of quantified Boolean formulae we denote with $\mathcal{C}^+$(resp. $\mathcal{C}^-$) the class of formulae in $\mathcal{C}$ whose free variables occur only positively (resp. negatively). For example, $\stcnf^-$ consists of all quantified Boolean formulae in prenex normal form with only existential quantifiers, where the quantifier-free part is in $\tcnf$ and the free variables occur only negatively. Note that in Boolean formulae all variables are free.

\paragraph*{Counting Problems and Counting Classes}

This paper aims to identify model-theoretic characterizations of counting classes in terms of team-based logics. Let us now recall relevant previous results on the descriptive complexity of counting problems.
We begin by defining the most important complexity classes for counting problems.
\begin{definition}
A function $f \colon \{0,1\}^* \rightarrow \mathbb{N}$ is in \cP if there is a non-deterministic polynomial time Turing machine $M$ \stfa inputs $x \in \{0,1\}^*$,
\begin{center}
$f(x)$ is the number of accepting computation paths of $M$ on input $x$.
\end{center}
\end{definition}
This definition can be generalized as follows.
\begin{definition} Let $\calC$ be a complexity class. 
  A function $f \colon \{0,1\}^* \rightarrow \mathbb{N}$ is said to be in $\#\cdot \calC$ if there are a language $L \in \calC$ and a polynomial $p$ such that for all $x \in \{0,1\}^*$:
  \[	f(x) = |\{y \mid |y| \leq p(|x|) \textrm{ and } (x,y) \in L\}|.\]
\end{definition}

Obviously $\cP= \#\cdot \P$, and it is well known that $\cP\subseteq \#\cdot  \NP \subseteq \#\cdot  \coNP = \cP^{\NP}$, where, under reasonable complexity-theoretic assumptions, all these inclusions are strict; see \cite{DBLP:journals/sigact/HemaspaandraV95} for a survey of these issues. 

\begin{definition}
A function $f \colon \{0,1\}^* \rightarrow \mathbb{N}$ is in \totP if there is a non-deterministic polynomial time Turing machine $M$ \stfa inputs $x \in \{0,1\}^*$,
\begin{center}
	$f(x)$ is the number of computation paths of $M$ on input $x$ minus $1$.
\end{center}
\end{definition}
Subtracting 1 from the number of computation paths is neccessary since otherwise \totP-functions could never map to 0. In \cite{DBLP:conf/mfcs/PagourtzisZ06} it was shown that $\totP$ is the closure \wrt parsimonious reductions of self-reducible counting problems from $\cP$ whose decision version is in $\P$. It follows that $\totP \subsetneq \cP$ unless $\P=\NP$.

Next, we define the relevant logical counting classes.
\begin{definition}\label{def-forel}
A function $f\colon \{0,1\}^* \rightarrow \mathbb{N}$ is said to be in $\cFO^\rel$ if there is a vocabulary $\sigma$ with a built-in linear order $\leq$, and an \FO-formula $\varphi(R_1, \dots, R_k, x_1, \dots, x_\ell)$ over $\sigma$ with free relation variables $R_1, \dots, R_k$ and free individual variables $x_1, \dots, x_\ell$ \stfa $\sigma$-structures $\calA$,
\begin{align*}
f(\enc_\sigma(\calA)) = |\{(S_1, \dots, S_k&, c_1, \dots, c_\ell):
  \calA \models \varphi(S_1, \dots, S_k, c_1, \dots, c_\ell\}|.
\end{align*}
If the input of $f$ is not of the appropriate form, we assume the output to be $0$.
\end{definition}
In the same fashion, subclasses of $\cFO^\rel$, such as $\cSigmaiRel{k}$ and $\cPiiRel{k}$ for arbitrary $k$, are defined by assuming that the formula $\varphi$ in the above definition is in the corresponding fragments $\Sigma_k$ and $\Pi_k$.

Recall the relationship between the above defined logical counting classes and \cP:

\begin{theorem}[\cite{DBLP:journals/jcss/SalujaST95}]
$\cSigmaiRel{0} = \cPiiRel{0} \subset \cSigmaiRel{1} \subset \cPiiRel{1} \subset \cSigmaiRel{2} \subset \cPiiRel{2} = \cFO^\rel = \cP$.\\
Furthermore, $\cSigmaiRel{0} \subseteq \classFont{FP}$.
\end{theorem}

Complete problems (i.e., hardest problems) for counting classes have also been studied extensively. Let us now recall three reductions that are relevant in this study. 
Let $f$ and $h$ be counting problems. We say that $f$ is \textit{parsimoniously} reducible to $h$ if there is a polynomial-time computable function $g$ such that $f(x)=h(g(x))$ for all inputs $x$, $f$ is \textit{Turing} reducible to $h$ if $f \in \fP^h$, and $f$ is \textit{metrically} reducible to $h$ if there are polynomial-time computable functions $g_1,g_2$ such that $f(x)=g_2(h(g_1(x)),x)$ for all inputs $x$. Clearly, metric reductions are  Turing reductions with one oracle query. Besides these three familiar reductions we now define another type of reductions, called {\em first-order reductions}. First recall that for any two vocabularies $σ_1, σ_2$, an {\em FO-interpretation} (or a {\em first-order query}) is a function $I : \struc[\sigma_1] \to \struc[\sigma_2]$, represented as a tuple $I = (\varphi_0, \varphi_{R_1}, \dots,\varphi_{R_\ell})$ of FO-formulae over $\sigma_1$ with $k$ free variables, that maps any structure $\mathcal{A} \in \struc[σ_1]$ to another structure $I(\mathcal{A}) \in \struc[σ_2]$, whose domain is a subset of $\textrm{dom}(\mathcal{A})^k$ (i.e., a set of $k$-ary tuples of elements in $\mathcal{A}$) defined by $\varphi_0$ and relations $R_i$ are defined by $\varphi_{R_i}$ (see \cite{DBLP:books/daglib/0095988} for detailed discussion). In the  team semantics case, we also need to define how teams  are transformed by the interpretation $I$.  The value $I(X)$ is defined in a straightforward way:  individual elements from $\calA^{k\cdot m}$ in $X$ are mapped  to elements from $I(\calA)^m$ in $I(X)$, where $k$ is the arity of tuples  in the domain of the structure $I(\calA)$, and $m$ is the arity of the team $I(X)$. Now, we  define first-order reductions via  FO-interpretations as follows:
\begin{definition}\label{fored}
	Let $f_1,f_2$ be two functions. We say $f_1$ is first-order reducible or FO-reducible to $f_2$ (denoted $f_1 \le_{\text{fo}} f_2$) if there are vocabularies $\sigma_1, \sigma_2$ with $\sigma_2=(R_1^{a_1}, \dots, R_l^{a_l})$ and an FO-interpretation $I = (\varphi_0, \varphi_{R_1}, \dots,\varphi_{R_\ell})$, where $\varphi_0, \varphi_{R_1}, \dots,\varphi_{R_\ell}$ are \FO-formulae over $\sigma_1$, such that for all $\calA_1 \in\struc[\sigma_1]$, there are $k \in \mathbb{N}$ and $\calA_2\in\struc[\sigma_2]$ with
	\[\textrm{dom}(\calA_2)=\{(x^1, \dots,x^k)\mid \calA_1 \models \varphi_0(x^1,\dots,x^k)\}\]
	and for all $i\le k$ 
	\[R_i((x^1_1,\dots, x_1^k), \dots, (x_{a_i}^1, \dots, x^{k}_{a_i})) \Leftrightarrow \calA_1 \models \varphi_{R_i}(x^1_1,\dots, x_1^k, \dots, x_{a_i}^1, \dots, x^{k}_{a_i})\]
	and $f_1(\enc_{\sigma_1}(\calA_1))=f_2(\enc_{\sigma_2}(\calA_2)).$
\end{definition}

It is often possible to find complete problems in counting classes by counting satisfying assignments for certain (quantified) Boolean formulae. Let $\mathcal{F}$ be a class of quantified Boolean formulae. Define the problem $\#\mathcal{F}$ as follows:
\problemdef{$\#\mathcal{F}$}{Formula $φ \in \mathcal{F}$}{Number of satisfying assignments of φ}
For example, $\#\mathrm{SAT}$, the function counting the number of satisfying assignments for propositional formulae, as well as its restriction $\#\tcnf$, are complete for \cP under parsimonious reductions, while $\#\bcnf^+$ and $\#\bcnf^-$ are complete for $\cP$ under Turing reductions. Observe that for all $\scnf$-formulae $\varphi$ it holds that the number of satisfying assignments is equal to that of the formula $\widetilde{\varphi}$  obtained by negating all literals in all clauses in $\varphi$. Thus, $\#\scnf^+$ and $\#\scnf^-$ are in a sense the same problem. In fact all our results for $\#\mathcal{C}^+$ (for a class of formulae $\mathcal{C}$) also hold for $\#\mathcal{C}^-$ and vice versa. This also holds for $\#\textnormal{Horn}$ and $\#\textnormal{DualHorn}$. Note that Aziz et al \cite{DBLP:conf/sat/AzizCMS15} studied the problem $\#\ssat$ under the name \emph{projected model counting} and observed that it is contained in $\#\cdot \NP$.

Next we introduce the central class for this paper, a class of counting problems in the context of team-based logics.
For any set $A$ of generalized dependency atoms, we define $\cTeam{\FO(A)}$ to consist of those functions counting non-empty satisfying teams for \mbox{$\FO(A)$-formulae}. Note that by the empty team property of dependence, independence, and inclusion logic formulae any function that counts all satisfying teams (including the empty team) could not attain the value 0.

\begin{definition}
	For any set $A$ of generalized atoms, $\cTeam{\FO(A)}$ is the class of all functions $f \colon \{0,1\}^* \to ℕ$ for which there is a vocabulary σ and an $\FO(A)$-formula $φ(\tu x)$ over σ with a tuple $\tu x$ of free first-order variables such that for all $\sigma$-structures $\calA$ with a built-in linear order $\leq$, addition $+$ and multiplication $\times$, 
\[f(\enc_σ(\calA)) = |\{X \in \textrm{team}(\calA, (\tu x)) : X\neq \emptyset \textrm{ and } \calA \models_X φ(\tu x)\}|.\]	

For inputs that do not encode a structure with the built-in predicates the function values are always 0.
We denote by $f_φ$ the function defined by φ.
\end{definition}

\begin{example}
As an example for how to work with team semantics in a counting context, we show that the $\cP$-complete problem $\#\bcnf^+$ is contained in both $\cTeam{\FO(\dep\dots)}$ and $\cTeam{\FO(\subseteq)}$.
Let $\varphi(x_1, \dots, x_n)= \bigwedge C_i \in \bcnf^+$, where each $C_i = \ell_{i,1} \vee \ell_{i,2}$ and $\ell_{i,j}\in\{x_1, \dots, x_n\}$. Consider the vocabulary $\tau_{\bcnf^+}=\{C^2\}$. We encode the formula $\varphi(x_1, \dots, x_n)$  by the structure $\mathcal{A}=(\{x_1, \dots, x_n\}, C^{\mathcal A})$, where $(x,y) \in C^{\mathcal A}$ \ifff the clause $x \vee y$ occurs in $\varphi$.

We show that $\#\bcnf^+ $ can be defined by counting non-empty teams (which correspond to assignments mapping at least one variable to true) satisfying suitable formulae from $\FO(\subseteq)$ as well as $\FO(\dep\dots)$. For this purpose, we encode Boolean assignments to the variables $x_1, \dots, x_n$ by teams over one variable $t$. Since the universe of our structures is exactly the set of variables of the formula $\varphi(x_1,\dots,x_k)$ in question, assignments to the variables can be encoded by inclusion of the values of the variables $x_i$ in $t$ in the team.

Now, the following $\FO(\subseteq)$-formula defines $\#\bcnf^+$:
\[φ_\subseteq(t) = \forall x \forall y (\neg C(x,y) \vee x \subseteq t \vee y \subseteq t).\]
Intuitively, this formula states that if a pair $(x,y)$ of variables in $\varphi$ occur in the same clause (i.e., $C(x,y)$ holds), then the value of one of the two variables $x,y$ is contained in the team, or it is set to true.

To define the same problem in $\FO(\dep\dots)$ where inclusion atoms are not any more available in the language, we need to encode assignments differently.
We now encode variables $x_i$ being set to 1 by not including them in the team over the variable $t$. The $\FO(\dep\dots)$-formula that defines $\#\bcnf^+$ is the following:
\begin{align*}
φ_{=\,\!(\dots)}(t) =& ∃\textrm{min}∀z \ \textrm{min} \leq z ∧\\
&∀x ∀y ∃x' ∃y' \Big( \dep{x,x'} ∧ \dep{y,y'} ∧ (\neg x=y ∨ x' = y') ∧\\
&\phantom{∀x ∀y ∃x' ∃y' \Big(} (x \neq t \lor x' = \textrm{min}) ∧ (\neg C(x,y) ∨ x' \neq \textrm{min} ∨ y' \neq \textrm{min})\Big).
\end{align*}
Intuitively, in the above formula we use existential quantifications together with dependence atoms to express that $x'$ is a function of $x$, and this function $f$ is guaranteed in the formula to be consistent with the assignment encoded by the team. We shall interpret a function mapping to the minimal element of the universe (encoded by the variable \textrm{min} in the formula) as an assignment to 0, and a function mapping to any other element as an assignment to 1. Now, the second last conjunct in our formula states that all variables that occur as values of $t$ in the team (which correspond to those $x_i$ set to 0) are mapped by our function $f$ to the minimal element. Finally, the last conjunct in our formula checks whether the $\bcnf$-formula is satisfied by the assignment encoded by $f$. Note that in order to talk about the assignment to two variables simultaneously, in the above formula we actually use two (equivalent) functions to encode the same assignment. 
\end{example}

\section{A Characterization of the Class $\#\cdot \NP$}

In this section, we characterize the class $\#\cdot \NP$ in terms of team-based logics. Our first result shows that $\#\cdot \NP$ is the largest class attainable by counting teams in team-based logics $\FO(A)$, as long as all generalized atoms in $A$ are NP-definable.

\begin{theorem}\label{IinNP}
	For any set $A$ of $\NP$-definable generalized atoms, $\cTeam{\FO(A)} \subseteq \#\cdot \NP.$
\end{theorem}

\begin{proofsketch}
  Let $φ(\tu x) \in \cTeam{\FO(A)}$. To show that $f_φ \in \#\cdot\NP$ we note that $f_φ$ can be computed by counting on input $\enc_\sigma(\mathcal{A})$ the number of teams $X$ such that $\mathcal{A} \models_X φ(\tu x)$. It is thus sufficient to show that the letter can be checked in \NP on input $(\enc_\sigma(\mathcal{A}), X)$.
In this proof, the only places that involve nondeterminism  are disjunctions (guess the split), existential quantifiers (guess the supplementing function) and \NP-definable generalized atoms (checkable in \NP by definition).
\end{proofsketch}

Next, we prove the converse inclusion of the above theorem by proving a stronger result: The whole class $\#\cdot \NP$ can actually be captured by a single generalized atom, the independence atom. 

\begin{theorem}\label{NPinI}
$\#\cdot \NP \subseteq \cTeam{\FO(⊥)} .$
\end{theorem}

\begin{proofsketch}
  It is sufficient to show $\#\cdot \NP \subseteq \#\Sigma_1^1$ since by Theorem~\hyperref[sig=ind]{\ref*{ind2sigma11} item \ref*{sig=ind}} we have $\#\Sigma_1^1 = \cTeam{\FO(⊥)}$.
  
  Let $f \in \#\cdot\NP$. Then there are a polynomial $p$ and $L \in \NP$ such that
  \[f(x) = |\{y \mid |y| = p(|x|), (x,y) \in L\}|.\]
  
  \noindent We  encode tuples $(x,y)$ of strings with $|y| = p(|x|)$ as structures $\mathcal{A}_{(x,y)}$ by encoding the string $x$ as a structure $\mathcal{A}_x$ in the standard way, and  $y$ as a unique relation $R_y$ over $\dom(\mathcal{A}_x)^k$ for some $k \in \mathbb{N}$.  Finally, Fagin's theorem  gives a $\Sigma_1^1$-sentence $φ$ such that for all $x$,

  \begin{align*}
  |\{y \mid |y| = p(|x|), (x,y) \in L\}| &= |\{y \mid \mathcal{A}_{(x,y)} \models φ\}|
  = |\{R \mid \mathcal{A}_x \models φ(R)\}|.
  \end{align*}
\end{proofsketch}

\begin{remark}
	The class \cP can also be characterized by counting teams. A variant $\mathcal{L}$ of dependence logic  that defines exactly the first-order definable team properties in the sense of Theorem \ref{ind2sigma11} was introduced in  \cite{JuhaFan}.  Since $\cP = \cFO$ (see \cite{DBLP:journals/jcss/SalujaST95}), this logic $\mathcal{L}$ captures \cP. 
	We do not present the details of $\mathcal{L}$ in this paper, but only note that $\mathcal{L}$ has weaker versions of quantifiers and disjunction instead of the standard ones as defined in Section \ref{sec:prelims}.
\end{remark}

\section{Counting Teams in Dependence and Inclusion Logic}

In this section, we study the smaller classes $\cTeam{\FO(\dep{\dots})}$ and $\cTeam{\FO(\subseteq)}$. 
We begin by showing that the $\#\cdot \NP$-complete problem $\#\scnf_*^-$, defined below, is contained in $\cTeam{\FO(\dep{\dots})}$. We will show $\#\cdot \NP$-completeness for $\#\scnf_*^-$ in Theorem \ref{thm:scnfstarNumNP}.

\problemdef{$\#\scnf_*^-$}{Formula $φ(x_1, \dots, x_k) \in \cnf^-$}{Number of satisfying assignments of φ, disregarding the all-0-assignment}
Note that the \emph{all-0-assignment} is the assignment mapping each variable to $0$.

\begin{theorem}\label{def-fodep}
	$\#\scnf_*^- \in \cTeam{\FO(\dep{\dots})}$.
\end{theorem}

We will show that the above problem is actually complete for $\cTeam{\FO(\dep{\dots})}$ with respect to first-order reductions. First-order reductions turn out to be particularly natural in our context, as all our classes are closed under these reductions.

\begin{theorem}\label{fo-closure}
	$\cTeam{\FO(A)}$ is closed under first-order reductions for $A \subseteq \{\dep{\dots}, \perp, \subseteq\}$.
\end{theorem}

Next we show that the problem $\#\scnf_*^-$ is hard and thus complete for $ \cTeam{\FO(\dep{\dots})}$ under first-order reductions. Our proof technique is similar to that of \cite{DBLP:journals/corr/DurandKRV15}, where the data complexity of inclusion logic is shown to be polynomial.

\begin{theorem}\label{dep_completeness}
	$\#\scnf_*^-$ is complete for $\cTeam{\FO(\dep{\dots})}$ \wrt first-order reductions.
\end{theorem}

Having proven our results for dependence logic $\FO(\dep{\dots})$, we now turn to inclusion logic $\FO(\subseteq)$. We first prove that $\cTeam{\FO(\subseteq)} $ is a subclass of $\cP$.

\begin{theorem}\label{inc1}
	$\cTeam{\FO(\subseteq)} \subseteq \cP$.
\end{theorem}

The above theorem naturally gives rise to the question whether $\cTeam{\FO(\subseteq)}$ actually coincides with \cP. However, we identify in the next lemma a particular property of $\cTeam{\FO(\subseteq)}$ functions, making this equivalence unlikely to hold.

\begin{lemma}\label{lem:inclLogDec}
Let $φ(\tu x) \in \FO(\subseteq)$ be a formula over a vocabulary σ. Then the language $L \defeq \{w \mid f_φ(w) > 0\}$ is in \P.
\end{lemma}

\begin{corollary}\label{inc2}
If $\P \neq \NP$, then $\cTeam{\FO(\subseteq)} \neq \cP$.
\end{corollary}

Theorem \ref{inc1} and Corollary \ref{inc2} indicate that $\cTeam{\FO(\subseteq)}$ is most likely a strict subclass of $\cP$. Nevertheless, we show in the next theorem that $\cTeam{\FO(\subseteq)}$ contains the problem $\#\textnormal{DualHorn}_*$ (same as $\#\textnormal{DualHorn}$, but without counting the all-0-assignment) which is complete for \cP \wrt Turing reductions. It is unknown whether $\#\textnormal{DualHorn}_* \in\#\FO(\dep{\dots})$.

\begin{theorem}\label{dhininc}
	$\#\textnormal{DualHorn}_*\in \#\FO(\subseteq)$.
\end{theorem}

We continue by exhibiting a hard problem for the class $\cTeam{\FO(\subseteq)}$. It is an open question whether the problem is definable by an inclusion logic formula.

\begin{theorem}\label{sdhhard}
	$\#\Sigma_1\textnormal{DualHorn}$ is hard for $\cTeam{\FO(\subseteq)}$ \wrt first-order reductions.
\end{theorem}

It seems to us that $\#\FO(\subseteq)$ is a strict subclass of $\cP$ and the decision versions of problems in $\#\FO(\subseteq)$ are in $\P$. For this reason, we now investigate  relationship between $\#\FO(\subseteq)$ and the subclass $\totP$. We show that $\#\FO(\subseteq)$ is a subclass of $\totP$ and that $\totP$ contains $\#\Sigma_1\textnormal{DualHorn}$. We conjecture that these classes do not coincide, but this question remains open.

\begin{theorem}\label{incintotp}
	$\#\FO(\subseteq) \subseteq \totP.$
\end{theorem}

From this result we can also conclude that $\#\FO(\dep{\dots}) \not\subseteq \#\FO(\subseteq)$. If that would be the case, then the problem $\#\scnf_*^+$ would be in the class \totP, which can not be true unless $\P = \NP$, since all problems in \totP have an easy (in polynomial time decidable) decision version. Note that the decision version of $\#\scnf_*^+$ is NP-complete.

\begin{theorem}\label{dhintotp}
	$\#\Sigma_1\textnormal{DualHorn} \in \totP.$
\end{theorem}

\section{Complete Problems for $\#\cdot \NP$}

In this section we show that $\#\scnf^-_*$ is $\#\cdot \NP$-complete. To this end, we first observe that $\#\scnf$ is $\#\cdot \NP$-complete. Afterwards we show that the smaller class $\#\scnf^-$ remains $\#\cdot \NP$-complete by adapting the proof for $\cP$-completeness of $\#\bcnf^+$ given by Valiant \cite{DBLP:journals/siamcomp/Valiant79}. We conclude this section with a reduction from $\#\scnf^-$ to $\#\scnf^-_*$ showing the $\#\cdot \NP$-completeness of the latter.

\begin{lemma}\label{lem:ssatComp}
	$\#\ssat$ and $\#\scnf$ are $\#\cdot \NP$-complete under parsimonious reductions.
\end{lemma}

\begin{theorem}\label{complete}
	$\#\scnf^-$ is $\#\cdot \NP$-complete under Turing reductions.
\end{theorem}

\begin{proofsketch}
  Membership follows from \ref{lem:ssatComp}, since $\#\scnf^-$ is a special case of $\#\scnf$.
	For the hardness proof, we show a chain of reductions adapted from the one used by Valiant \cite{DBLP:journals/siamcomp/Valiant79} to show the \cP-completeness of $\#\bcnf^+$. 	Recall that the main steps of Valiant's chain of reductions are:
	\begin{alignat*}{4}
	\#\tcnf &\leq \perm &&\leq \#\perfmat\\
	&\leq \#\impmat	&&\leq \#\bcnf^+.
	\end{alignat*} 
	Our idea is to add a $\stcnf$-formula to the input of each problem in the above chain of reductions, and to express certain properties of the respective inputs in the added formula. We then count only the solutions to the input that also satisfy the added formula.
  
  As part of the chain of reductions we will make use of the  following problem:
  \problemdef{$\#(\tcnf, \stcnf^-)$}{Formula $φ(x_1, \dots, x_k) \in \tcnf$ and  formula $ψ(x_1, \dots, x_k) \in \stcnf^-$}{Number of satisfying assignments of $φ ∧ ψ$}
We will reduce $\#\stcnf$ to $\#(\tcnf, \stcnf^-)$, and then apply the above chain of reductions (with added formulae, as described above). This results in a reduction to $\#(\bcnf^-, \stcnf^-)$, defined analogously to the above problem. Finally it  is straightforward to show $\#(\bcnf^-, \stcnf^-) \leq \#\scnf^-$ using the fact that for $φ \in \bcnf^-$ and $ψ \in \stcnf^-$, the prenex normal form of $φ ∧ ψ$ is a $\stcnf^-$-formula. We conclude by sketching the first reduction.
 
  \underline{$\#\stcnf \leq \#(\tcnf, \stcnf^-)$}: We construct for any $φ \in \stcnf$ two formulae $φ' \in \tcnf$ and $ψ \in \stcnf^-$ such that $\#\stcnf(φ) = \#(\tcnf, \stcnf^-)(φ', ψ)$.

Assume $φ = ∃y_1 \dots ∃y_\ell \bigwedge{C_i} ∧ \bigwedge{D_i} ∧ \bigwedge{E_i}$, where clauses $C_i$ only contain free variables of φ, clauses $D_i$ contain only bound variables of φ, and clauses $E_i$ contain at least one free and at least one bound variable of φ. We can now simply add all clauses $C_i$ to $φ'$ and all clauses $D_i$ to ψ.
  
  To handle the remaining clauses, we add for each $E_i$ a new free variable $e_i$. We then express in $φ'$ that $e_i$ is true \ifff clause $E_i$ is not satisfied by the assignment to the free variables, and express in ψ that $E_i$ has to be satisfied by the assignment to the bound variables if $e_i$ is true. The former does not involve any bound variables and for the latter, the only needed free variable is $e_i$, which only occurs negatively.
  
\end{proofsketch}

Because of the special role of the empty team in the team logics we consider, we will also show the completeness for another version of $\#\scnf^-$, denoted as $\#\scnf_*^-$,  for which the all-0-assignment is not counted.

\begin{theorem}\label{thm:scnfstarNumNP}
	The problem $\#\scnf_*^-$ is $\#\cdot \NP$-complete under Turing reductions.
\end{theorem}

\section{Conclusion}

 In this paper we have studied the following hierarchy of classes defined by counting problems for team-based logics:
\[\begin{array}{cccc}
	\totP & \subseteq \cTeam{\mathcal{L}} = \cP  & \subseteq & \cTeam{\FO(\bot)} = \#\cdot\NP  \\
	\rotatebox{90}{$\subseteq$} &   &   & \rotatebox{90}{$\subseteq$} \\
	\cTeam{\FO(\subseteq)}       &  &    & \cTeam{\FO(\dep{\dots})}
\end{array}\]
We also showed that our classes are closed under first-order reductions and that $\cTeam{\FO(\subseteq)}$ and $\cTeam{\FO(\dep{\dots})}$ contain complete problems from \cP and $\#\cdot\NP$, respectively. The latter problem is even complete for $\cTeam{\FO(\dep{\dots})}$ under first-order reductions.

The connection between $\cTeam{\FO(\dep{\dots})}$ and the classes \cP and $\#\cdot \NP$ is not yet clear. While we know that a complete problem from $\#\cdot \NP$ is contained in it, it is open whether the class coincides with $\#\cdot \NP$, and (if not) whether it contains the class \cP. We conjecture that the answer to both questions is negative, since the defining logic has closure properties that make it unlikely to contain counting versions of non-monotone problems from \cP. 

Regarding $\cTeam{\FO(\subseteq)}$, the search for a complete problem could be interesting. We have only showed that the problem $\#\textnormal{DualHorn}$ is contained in this class and the problem $\#\Sigma_1\textnormal{DualHorn}$ is hard for this class, but neither of the problems is known to be complete.

The lower end of our hierarchy deserves further study as well. The class $\cTeam\FO$ (\ie, the class with no  dependency atoms  in the formulae) can be shown to be a subclass of $\FTC$, the class of functions computable by families of polynomial size constant depth majority circuits (see \cite{DBLP:books/daglib/0097931}). The circuit-based counting class \cAC, counting proof trees in polynomial size constant depth unbounded fan-in circuits \cite{DBLP:books/daglib/0097931}, was characterized in a model-theoretic manner by counting assignments to free function symbols in certain quantifier-restricted FO-formulae \cite{DBLP:conf/csl/0001HKV16}. A similar quantifier restriction for $\cTeam{\FO(A)}$, where $A$ consists of the dependency atom plus a totality atom (that we did not study  in the present paper), also leads to a characterization of \cAC. This suggests that 
low level counting classes and circuit classes in the context of counting problems for team-based logics might be worth studying. 
Another natural question is to search for generalized dependency atoms that lead to interesting relations to complexity classes. Besides the aforementioned totality atom, the constancy or the exclusion \cite{DBLP:journals/apal/Galliani12} atom are worth examining. In particular, it is an open question to find an atom $A$ such that $\cTeam{\FO(A)}=\cP$. The logic $\mathcal{L}$ of \cite{JuhaFan}, though it satisfies the equality, is not of this form.

In the context of counting complexity theory, an interesting question to study is the approximability of problems in different classes. In our case, it is unlikely that any of our classes is efficiently approximable (in the sense of FPRAS): In \cite{DBLP:journals/jcss/DyerGJ10} it was shown that it is unlikely that the number of satisfying assignments of CSPs can be approximated by an FPRAS, unless all relations in the constraint language are affine. Since our classes contain counting problems for classes of formulae which do not admit this property, this result  applies. Consequently, it would be interesting to study restrictions of our full classes to obtain, possibly, efficiently approximable fragments.

Our proof of the completeness of $\#\scnf^-$ for $\#\cdot \NP$ introduces problems that arise from ``pairing decision problems'' and gives simultaneous reductions between such pairs. This idea might be helpful in other contexts as well; in particular it should lead to more interesting complete problems for $\#\cdot \NP$ or higher levels $\#\cdot\Sigmai{k}$ of the counting polynomial-time hierarchy.

\bibliography{ref}

\appendix

\section{Appendix}

\begin{proof}[Proof of Theorem \ref{IinNP}]
	Let $φ(x_1, \dots, x_k) \in \FO(A)$.  For a given input structure $\mathcal{A}$, to count the number of (non-empty) teams $X$ with $\mathcal{A} \models_X φ$ with a $\#\cdot \NP$-algorithm, we first non-deterministically guess a team $X$ and check in $\NP$ whether $\mathcal{A} \models_X φ$ holds.
	
	For the latter note that for any fixed formula φ without disjunctions and existential quantifiers, $\mathcal{A} \models_X φ$ can be checked in nondeterministic polynomial time (where nondeterminism is only needed to handle generalized atoms). Disjunctions can be handled by non-deterministically guessing the subteams for the disjuncts, and existential quantifiers can be handled by non-deterministically guessing the supplementing functions. Hence, $\mathcal{A} \models_X φ$ can be checked in \NP.
\end{proof}

\begin{proof}[Proof of Theorem \ref{NPinI}]
	First note that 
	$\#\Sigma_1^1 = \cTeam{\FO(⊥)}$,
	because by Theorem~\hyperref[sig=ind]{\ref*{expr} item \ref*{sig=ind}}, any sentence $\varphi(R)\in \Sigma_1^1$ with  a $k$-ary relation symbol $R$ can  be easily turned into a sentence $\varphi'(R')$ for some $(k+1)$-ary $R'$ such that $\varphi$ and $\varphi'$ define the same functions and $\varphi'(R')$ is only satisfied by non-empty relations. It then suffices to show that $\#\cdot \NP \subseteq \#\Sigma_1^1$.
	
	Consider the string vocabulary $\tString = (S^1)$ used to encode binary strings as first-order structures. For any binary string $w = w_0 w_1 \dots w_{n-1} \in \{0,1\}^*$, we define the structure encoding $w$ as $\mathcal{A}_w = (\{0, 1, \dots, n-1\}, \textbf{S})$, where $\textbf{S}(i) = w_i$ for all $i$. Note that with this definition the standard binary encoding of the structure $\mathcal{A}_w$ is the string $w$ itself, namely $\enc_{\tString}(\mathcal{A}_w) = w$ for all $w \in \{0,1\}^*$.
	
	Let $f \in \#\cdot \NP$ via $L \in \NP $ and the polynomial $p$, with $p(n)= n^\ell+c$ for some $\ell, c \in \mathbb{N}$. By definition we have $f(x)=|\{y:|y| = p(|x|), (x,y)\in L)\}|$. We encode tuples $(x,y)$ with $x \in \{0,1\}^*$ and $y \in \{0,1\}^{|x|^k}$ by structures over the vocabulary $τ_k = \tString \cup (R^k)$. Note that such an encoding is possible because we can define the extension of the numerical predicates to tuples in \FO (see \cite{DBLP:books/daglib/0095988}). For strings $x$ with $|x| \geq 2$ we can choose $k$ such that $|x|^k \geq p(|x|)$ (strings of length 1 can be handled separately). Fix such a $k$. We denote by $\mathcal{A}_{(x,y)}$ the $τ_k$-structure encoding the tuple $(x,y)$. 
	
	Now consider the language
	\[L' \defeq \left\{\enc_{τ_k}(\mathcal{A}_{(x,y)}) \middle|
	\begin{array}{c}
	\mathcal{A}_{(x,y)} \in \struc[τ_k],
	y = y_0 \dots y_{|x|^k -1},\\
	y_{p(|x|)} = \dots = y_{|x|^k -1} = 0
	\textrm{ and } (x,y_0 \dots y_{p(|x|)-1}) \in L
	\end{array}\right\}.\]
	For any given $x$, $\enc_{τ_k}(\mathcal{A}_{(x,y)})$ is an element of $L'$ if and only if the first $p(|x|)$ bits of $y$ form an input $z$ such that $(x,z) \in L$ and the rest of the bits are fixed to be 0. Thus,
	\[f(x) = |\{y : \enc_{τ_k}(\mathcal{A}_{(x,y)}) \in L'\}|.\]
	Obviously $L' \in \NP$, which, by Fagin's Theorem (see \cite{fagingeneralized}), implies that there is a sentence $φ \in \Sigma_1^1$ over $τ_k$ such that
	\[\enc_{τ_k}(\mathcal{A}_{(x,y)}) \in L' ⇔ \mathcal{A}_{(x,y)} \models φ.\]
	Viewing φ as a formula over the vocabulary \tString with free relational variable $R$ of arity $k$ we have
	\[\mathcal{A}_x \models φ(R) ⇔ \enc_{τ_k}(\mathcal{A}_x, R) \in L'\]
	for all $x \in \{0,1\}^*$, which yields
	\[f(x) = |\{R : \mathcal{A}_x \models φ(R)\}|.\]
	Hence $f \in \#\Sigma_1^1$.
\end{proof}

For the following proofs we need an encoding of of $\scnf^-$-formulae as first-order structures. For this, we use the vocabulary
\[τ_{{\scnf}^-} = (F^1, B^1, P^2, N^2),\]
where the predicates $F$ and $B$ state, respectively, which variables occur free and which occur bound in the encoded formula, and the predicate $P$ (resp. $N$) is the incidence relation between clauses and positive (resp. negated) variables. More precisely, an arbitrary $\scnf^-$-formula $\varphi(x_1,\dots, x_k)=\exists y_1 \dots \exists y_l \psi(x_1, \dots, x_k, y_1, \dots, y_l)$ with $\psi=\bigwedge_{i=1}^{n}C_i$, $C_i=l_{i,1}\vee \dots \vee l_{i,m_i}$ and
\[l_{i,j}\in \{\neg x_i| i \in \{1, \dots, k\}\} \cup \{y_i,\neg y_i| i \in \{1, \dots, l\}\}\]
is encoded as the $τ_{\scnf^-}$-structure
\[\mathcal{A} = (\textrm{dom}(\mathcal{A}), F^\mathcal{A}, B^\mathcal{A}, P^\mathcal{A}, N^\mathcal{A})\]
defined as follows: The elements of $\textnormal{dom}(\mathcal{A})$ are numerical encodings of the variables and clauses in φ. With some abuse of notations, we write $\textrm{dom}(\mathcal{A}) = \{x_1, \dots, x_k, y_1, \dots, y_l, C_1, \dots, C_n\}$, where we identify the variables and clauses with their encodings. The interpretations of the predicate symbols are defined as:
\begin{itemize}
	\item $F^\mathcal{A}(x)$ \ifff $x$ is a free variable in φ
	\item $B^\mathcal{A}(x)$ \ifff $x$ is a bound variable in φ
	\item $P^\mathcal{A}(C_i,x)$ \ifff there is a $j$ such that $l_{i,j}=x$
	\item $N^\mathcal{A}(C_i,x)$ \ifff there is a $j$ such that $l_{i,j}=\neg x$
\end{itemize}

\begin{proof}[Proof of Theorem \ref{def-fodep}]
	By Theorem~\hyperref[dep=esor-]{\ref*{expr} item \ref*{dep=esor-}}, it suffices to construct a $\Sigma_1^1$-formula $ψ(T)$ with $T$ occurring only negatively such that for each $τ_{\scnf}$-structure $\mathcal{A}$, the number of relations $\textbf{T}$ with $\mathcal{A} \models ψ(\textbf{T})$ is equal to the number of satisfying assignments of the $\scnf^-$-formula encoded by $\mathcal{A}$. Note a subtle point in this setting: The all-0-assignment is not counted by $\#\scnf_*^-$. In the formula $ψ(T)$, the all-0-assignment corresponds to the empty relation, which, in turn, corresponds to the empty team in the translation to $\FO(\dep{\dots})$ given in Theorem~\hyperref[dep=esor-]{\ref*{expr} item \ref*{dep=esor-}}. We did not count the empty team when we defined the functions $f_ψ$ for $\#\FO(\dep{\dots})$, and the all-0-assignment is thus not counted in this setting.
	
	Let us now define the formula $ψ(T)$. First, let $\psi_1$ be an \FO-formula expressing that $\mathcal{A}$ is a correct encoding of a $\scnf^-$-formula (i.e., $F$ and $B$ correspond to disjoint sets, free variables occur only negatively in the clauses, etc.). We omit here the precise definition of $\psi_1$. Let $\psi_2$ be a formula expressing that $T$ only assigns values to free variables and that each clause must be satisfied by the assignment, that is, formally,
	\begin{align*}
	\psi_2&(T)= \exists S \forall C
	\bigg(\Big(\neg F(C) \wedge \neg B(C)\Big)\rightarrow 
	\\
	& \exists x \Big(N(C,x)\wedge \big((B(x)\wedge \neg S(x))
	\vee (F(x)\wedge \neg T(x))\big)\Big) 
	 \vee\exists x \Big(P(C,x)\wedge (B(x)\wedge S(x)\Big)\bigg), 
	\end{align*}
	Finally, let
	\[ψ(T) = ψ_1 ∧ ψ_2(T)\]
	It is easy to see that the formula $\psi(T)$ has the desired properties.
\end{proof}

\begin{proof}[Proof of Theorem \ref{fo-closure}]
	
		Let $f_1,f_2$ be two functions with $f_1\le_{\text{fo}} f_2$ and let $\sigma_1,\sigma_2,I$ be defined as in definition \ref{fored}. If $f_2\in\cTeam{\FO(A)}$, then there is a formula $\varphi(x_1, \dots, x_n) \in \FO(A)$ such that $f_2(\enc_{\sigma_2}(\calA))$ is equal to the number of teams $X$ satisfying $\calA \models_X \varphi(x_1,\dots,x_n)$. 
		
		We define a formula $\psi(x_1^1,\dots, x_n^k)\in\FO(A)$ over $\sigma_1$ such that for all $\calA \in \struc[\sigma_1]$ and all teams $X$ over $x_1^1,\dots, x_n^k$, 
		\[\calA \models_{X}\psi(x_1^1,\dots, x_n^k) \Leftrightarrow I(\calA) \models_{I(X)} \varphi(x_1,\dots,x_n) \]
	
    We first define  inductively a formula $φ'$ from $\varphi$ as follows:
		
		\begin{itemize}
			\item $φ'=\varphi_R(x_1^1,\dots,x_1^k, \dots, x_{a_l}^1, \dots, x_{a_l}^k)$ if $\varphi=R(x_1,\dots, x_{a_l})$
			\item $φ'=\widetilde{\varphi}_R(x_1^1,\dots,x_1^k, \dots, x_{a_l}^1, \dots, x_{a_l}^k)$ if $\varphi=\neg R(x_1,\dots, x_{a_l})$
			\item $φ'=\bigwedge_i x^i=y^i$ if $\varphi= (x=y)$
			\item $φ'=(x^1, \dots, x^k)\le (y^1, \dots, y^k)$ if $\varphi= x\le y$
			\item $φ'=+((x^1, \dots, x^k),(y^1, \dots, y^k),(z^1, \dots, z^k))$ if $\varphi= +(x,y,z)$
			\item $φ'=\times((x^1, \dots, x^k),(y^1, \dots, y^k),(z^1, \dots, z^k))$ if $\varphi= \times(x,y,z)$
			\item $φ'=\bigwedge_i =(x_1^1, \dots, x_1^k, \dots, x_m^1, \dots, x_m^k,y^i)$ if $\varphi=\ \dep{x_1, \dots, x_m, y}$
			\item $φ' = (x_1^1, \dots x_1^k, \dots, x_m^1, \dots, x_m^k) \subseteq (y_1^1, \dots, y_1^k, \dots, y_m^1, \dots, y_m^k)$\\
			if $φ = (x_1, \dots, x_m) \subseteq (y_1, \dots, y_m)$
			\item $φ' = (y_1^1, \dots, y_1^k, \dots, y_{m_1}^1, \dots, y_{m_1}^k) \bot_{(x_1^1, \dots, x_1^k, \dots, x_{m_2}^1, \dots, x_{m_2}^k)} (z_1^1, \dots, z_1^k, \dots, z_{m_1}^1, \dots, z_{m_1}^k)$\\
			if $φ = (y_1, \dots, y_{m_1}) \bot_{(x_1, \dots, x_{m_2})} (z_1, \dots, z_{m_1})$
			\item $φ'=φ_1' \wedge φ_2'$ if $φ = φ_1 \land φ_2$
			\item $φ'=φ_1' \vee φ_2'$ if $\varphi=\varphi_1 \vee \varphi_2$
			\item $φ'=\exists x^1 \exists x^2 \dots \exists x^k \varphi_0(x^1, \dots, x^k) \wedge φ_1'$ if $\varphi=\exists x \varphi_1 $
			\item $φ'=\forall x^1 \forall x^2 \dots \forall x^k \widetilde{\varphi}_0(x^1, \dots, x^k) \vee φ_1'$ if $\varphi=\forall x \varphi_1 $
		\end{itemize}
		
\noindent where for any $\FO$-formula $ξ$, $\widetilde{ξ}$ is the formula $\neg ξ(x_1^1,\dots,x_1^k, \dots, x_{a_l}^1, \dots, x_{a_l}^k)$ in negation normal form. The numerical predicates $\le,+,\times$ are definable for tuples by Immerman \cite{DBLP:books/daglib/0095988}.
		Now, define  $\psi(x_1^1,\dots, x^k_n)=φ'(x_1^1,\dots, x_n^k) \wedge \bigwedge_i \varphi_0(x_i^1,\dots,x_i^k)$. 	
\end{proof}

\begin{proof}[Proof of Theorem \ref{dep_completeness}] By Theorem \ref{def-fodep}, $\#\scnf_*^-$ is contained in $\cTeam{\FO(\dep{\dots})}$. It remains to show hardness. Let $\varphi(x_1,\ldots ,x_m)\in \FO(\dep{\dots})$, and let $\mathcal A$ be a structure with domain $A=\{0,\ldots,n-1\}$. We reduce computing the value of $f_\varphi(\mathcal{A})$ to counting the number of satisfying assignments (apart from the all-0-assignment) of a suitable Boolean formula $\Gamma^{φ,\mathcal{A}} \in \scnf^-$. By \cite{DBLP:books/daglib/0030191}, we may assume without loss of generality that $\varphi$ is of the form
	\[ \forall y_1\ldots \forall y_k \exists y_{k+1}\ldots \exists y_{k+l} ( \bigwedge_t \dep{\tu{u}_t,w_t} \wedge \theta ),   \]
	where $\theta$ is a quantifier-free \FO-formula, $w_t\in \{y_{k+1},\ldots , y_{k+l} \}$, and $\tu{u}_i$ is a tuple consisting of some of the variables $y_1,\ldots, y_k$. $\Gamma^{φ,\mathcal{A}}$ is a formula over the set $\{X_s\mid s\in A^{m}\cup A^{m+1}\cup\dots\cup A^{m+k+l} \}$ of propositional variables. Observe that each such $s$ can be identified with a partial first-order assignment over the domain $\{x_1,\ldots,x_m,y_1,\ldots,y_{k+l}\}$. Clearly, the number of such assignments and consequently also the number of variables $X_s$ is polynomial.
	The variables $X_s$ for $s \in A^{m+1}\cup \cdots \cup A^{m+k+l}$ will be existentially quantified in $\Gamma^{φ,\mathcal{A}}$, whereas the variables $X_s$ for $s \in A^{m}$ will remain free and occur only negatively. 
	
	We now define the set $\mathfrak{C}$ of clauses of $\Gamma^{φ,\mathcal{A}}$.
	For every universally quantified variable $y_i$ in $\varphi$, and for every $s\in A^{m+(i-1)}$ we introduce to $\mathfrak{C}$ the following set of clauses:
	\[ \{ X_s\rightarrow X_{s'} \mid s'\in A^{m+i} \textrm{ and } s=s'\upharpoonright \{(\tu{x},y_1,\ldots,y_{i-1})\}\}\tag{$\forall$}\label{forall} \]
	For every existentially quantified variable $y_i$ in $\varphi$, and for every $s\in A^{m+(i-1)}$ we introduce the following clause:
	\[ X_s\rightarrow \hspace{30pt}\mathclap{\bigvee_{\substack{s'\in A^{m+i} \textrm{ and }\\s=s'\upharpoonright \{(\tu{x},y_1,\ldots,y_{i-1})\}}}} \hspace{30pt} X_{s'} . \tag{$\exists$}\label{exists}      \]
	The quantifier-free part of the formula $\varphi$ also gives rise to clauses as follows. For each dependence atom $\dep{\tu{u}_t,w_t}$ we introduce the set of clauses:
	\begin{align*}
	\{ \neg X_s\vee \neg X_{s'} \mid s,s' \in & A^{m+k+l}, 
	 s(\tu{u}_t)=s'(\tu{u}_t)  \textrm{ and }   s(w_t)\neq s'(w_t)\}.
	\tag{$\dep{\dots}$}\label{dep}
	\end{align*}
	Finally, for the \FO-formula $\theta$ the team semantics satisfaction condition stipulates that all assignments  $s \in  A^{m+k+l}$ should satisfy $\theta$ (since $\theta$ is flat). This can be expressed by introducing the following two sets of clauses:
	\[  \{ X_s\rightarrow \top \mid s\in A^{m+k+l}  \textrm{ and } \mathcal{A}\models_s \theta \}\tag{$\top$}\label{top} \]
	and 
	\[  \{ X_s\rightarrow \bot \mid s\in A^{m+k+l}  \textrm{ and } \mathcal{A}\not \models_s \theta \}. \tag{$\bot$}\label{bot}\]
	
	Now define $\Gamma^{φ,\mathcal{A}} \in \scnf^-$ with respect to the input $\varphi$ and $\mathcal{A}$ as
	\[\Gamma^{\varphi, \mathcal{A}} := \exists \{X_s\mid s\in \bigcup_{1\le i\le k+l} A^{m+i}\}\bigwedge \mathfrak{C}.  \]
	
	Clearly there is a 1-1-correspondence between teams $X$ over domain $\{x_1, \dots, x_m\}$ and assignments $S$ of formula $\Gamma^{\varphi, \mathcal{A}}$. Furthermore it is also easy to check that for all teams $X$ with domain $\{x_1,\ldots ,x_m\}$
	\[\mathcal A\models_X\varphi(x_1,\ldots ,x_m)\Leftrightarrow S\models \Gamma^{\varphi, \mathcal{A}},\]
	where the Boolean assignment $S$ is defined as $S(X_s)=1$ \ifff $s\in X$.

	We now argue that the above can be done using a first-order reduction. First note that clauses of types \ref{forall},\ref{exists},\ref{dep} only depend on the input length $n$. We thus only need to give a closed formula in $n$ (only using $+$ and $\times$) to calculate for every clause the incident literals. 
	
	As an illustration, we give the details for clauses of type \ref{forall}. In total, there are $\sum_{i=0}^{k} n^{m+i}$ clauses of this type, where each summand gives the number of clauses corresponding to one universal quantifier. Consider  a clause $C$ that is counted in the $i$-th summand. $C$ is of the form $X_s \rightarrow X_{s'}$, where $s' \in A^{m+i}$ and $s=s'\upharpoonright \{(\tu{x},y_1,\ldots,y_{i-1})\}$. Let $j$ be the position of $C$ within the $i$-th summand. Then $s'$ is the $j$-th element of $A^{m+i}$ and $s$ is defined accordingly. Clauses arising from existential quantifiers can be handled analogously.

	For clauses corresponding to dependence atoms it is helpful to simplify the construction as follows: 
	Consider the atom $\dep{\tu{u}_t, w_t}$. Instead of adding the clauses from (\ref{dep}) we add a clause for each pair $s,s'\in A^{m+k+l}$ and check whether that pair falsifies the additional condition using an $\FO$-formula. If it does, we construct a fixed tautology instead of the clause $(\neg X_s \vee \neg X_{s'})$.
	
	Clauses of type \ref{top},\ref{bot} depend not only on the input length but on the  actual input. They can be handled simultaneously by adding clauses of the form $X_s \rightarrow X_s$ for type \ref{top} and $X_s\rightarrow \bot \equiv \neg X_s$ for type \ref{bot}. Distinguishing between those two types only depends on the $\FO$-definable property $\calA \models_s \theta$. This finishes the proof.
\end{proof}

\begin{proof}[Proof of Theorem \ref{inc1}]
		To count for a given input structure $\mathcal{A}$ the number of satisfying teams for a formula in $\FO(\subseteq)$, we simply guess a team and verify that it satisfies the formula. The latter step can be done in polynomial time, since model-checking for $\FO(\subseteq)$ is in $\P$ by Theorem~\hyperref[inc=lfp]{\ref*{expr} item \ref*{inc=lfp}}.
\end{proof}

\begin{proof}[Proof of Lemma \ref{lem:inclLogDec}]
	Let $w \in \{0,1\}^*$ and $\calA(w) \defeq \enc_σ^{-1}(w)$. The condition $f_φ(w) > 0$ asks whether there is a non-empty team $X \in \textrm{team}(\calA(w), \tu x)$ such that $\calA(w) \models_X φ(\tu x)$, which is equivalent to asking whether $\calA \models_{\{\emptyset\}} ∃\tu x φ(\tu x)$ holds.
	By Theorem~\hyperref[inc=lfp]{\ref*{expr} item \ref*{inc=lfp}}, over ordered structures, the properties definable by $\FO(\subseteq)$-sentences are exactly the properties in P, hence it follows that $L \in \P$.
\end{proof}

\begin{proof}[Proof of Corollary \ref{inc2}]
	Suppose $\cTeam{\FO(\subseteq)} = \cP$. This means that $\#3\textrm{CNF}$ is contained in $\cTeam{\FO(\subseteq)}$. Then, by Lemma \ref{lem:inclLogDec} the language $\{w \mid \#3\textrm{CNF}(w) > 0\} = 3\textrm{CNF}$ is in \P, which implies $\P = \NP$, contradicting the assumption.
\end{proof}

To represent DualHorn-formulae as first-order structures we use the vocabulary $\tau_{\textnormal{DualHorn}}=(C^1,P^2,N^2)$ with the following interpretations of the predicates:
\begin{itemize}
	\item $C(x)$ \ifff $x$ is a clause in φ
	\item $P(C_i,x)$ \ifff there is a $j$ such that $l_{i,j}=x$
	\item $N(C_i,x)$ \ifff there is a $j$ such that $l_{i,j}=\neg x$
\end{itemize}

\begin{proof}[Proof of Theorem \ref{dhininc}]
	By Theorem \hyperref[myopic]{\ref*{expr} item \ref*{myopic}} it suffices to give a myopic $\tau_{\textnormal{DualHorn}}$-formula $\psi$ such that for all DualHorn-formulae $\varphi$ the number of satisfying assignments is equal to the number of relations $\textbf{R}$ such that $ (\calA, \textbf{R}) \models \psi(R)$. 
	The  formula $\psi$ is defined as
	\begin{align*}
	\psi(R)= \forall x (R(x) \rightarrow (\forall C\ &((\neg \exists z N(C,z))\rightarrow (\exists y P(C,y)\wedge R(y))) \\
	&\wedge (N(C,x)\rightarrow (\exists y P(C,y)\wedge R(y)))))
	\end{align*}
	
	Now suppose $R$ satisfies the formula $\psi$. Let $x \in R$. It follows that all clauses that contain $x$ or contain only positive literals are satisfied by $R$: If $x$ is positively contained in a clause $C$, then it is already satisfied since $x \in R$. If $x$ is negatively contained in $C$, then there must be another variable $y$ that occurs positively in $C$ (since each clause contains at most one negative literal) with $y \in R$. If $C$ only contains positive literals, then there must be one $y \in R.$ This only works if there is at least one variable included in $R$. If $R$ is empty in the first place the premise of the first implication is always false and therefore the conclusion can be anything. It follows that $\psi(\emptyset)$ is always true, which is no surprise since it is a myopic formula. But since we are only looking for non-empty relations (assignments $\beta$ that are not the all-0-assignment respectively) this is not a problem.
	Now for all assignments $\beta\neq \emptyset$ it holds that $\beta \models \varphi \iff \calA_{\varphi}, \beta \models \psi(\beta)$.
\end{proof}

\begin{proof}[Proof of Theorem \ref{sdhhard}]
	The proof is analogous to that of Theorem \ref{dep_completeness} (see also \cite{DBLP:journals/corr/DurandKRV15}). Note that for inclusion logic there is a normal form, introduced in \cite{Yang2019}, very similar to the one used in the proof of Theorem \ref{dep_completeness}. Thus, the only remaining ingredient needed for the proof is the fact that inclusion atoms $\tu{x} \subseteq \tu{y}$ can be expressed by adding the following type of DualHorn clauses:
	\[   X_s\rightarrow  \hspace{20pt} \mathclap{\bigvee_{\substack{s'\in M^{m+k+l} \textrm{ and }\\s(\tu{x})=s'(\tu{y})}}} \hspace{20pt} X_{s'}  . \]
\end{proof}

\begin{proof}[Proof of Theorem \ref{incintotp}]
	In \cite{Gradel2016} Grädel showed that given a formula $\varphi \in \FO(\subseteq)$ and a team $X$, the maximal subteam of $X$ that satisfies $\varphi$ is unambiguous and can be computed in polynomial time. We now use this result to construct a $\totP$ machine for $\#\FO(\subseteq)$.
	
	Let $X_{\text{max}}$ be the maximal team satisfying $\varphi$ and $X_{\text{full}}$ the team containing all assignments. The machine first computes $X_{\text{max}}$ by computing the maximal satisfying subteam of $X_{\text{full}}$. If $X_{\text{max}}=\emptyset$, the machine halts. If $X_{\text{max}}\neq\emptyset$ the machine computes the maximal subteam of $X_{\text{max}}\setminus \{s\}$ for all $s \in X_{\text{max}}$. For every maximal subteam that is not the empty team the machine branches. The computation then continues recursively on the branches of those nonempty subteams treating them as the new maximal subteam.
\end{proof} 

\begin{proof}[Proof of Theorem \ref{dhintotp}]
	Let $\varphi \in \Sigma_1\textnormal{DualHorn}$. The machine that witnesses membership in $\totP$ works as follows:
	Choose a free variable, assign the value $0$ to it and simplify the formula (remove all clauses that are already satisfied and all positive occurrences of the variable). Compute in polynomial time if the resulting formula is satisfiable and do the same for the value $1$ afterwards. Create one branch for each of those two formulae that is satisfying. In each branch do this recursively for the next variable that is not assigned yet. 
\end{proof}

\begin{proof}[Proof of Lemma \ref{lem:ssatComp}]
	Membership of $\#\ssat$ in $\#\cdot \NP$ is due to Aziz et al \cite{DBLP:conf/sat/AzizCMS15}. Since $\#\scnf$ is a restriction of $\#\ssat$, membership for $\#\scnf$ also follows immediately.
	A simple adaptation of Cook's proof of $\NP$-completeness of SAT \cite{DBLP:conf/stoc/Cook71} shows that both problems are hard for $\#\cdot \NP$ \wrt parsimonious reductions.
\end{proof}

\begin{proof}[Proof of Theorem \ref{complete}]
	Membership follows from \ref{lem:ssatComp}, since $\#\scnf^-$ is a special case of $\#\scnf$.
	For the hardness proof, we show a chain of reductions adapted from the one used by Valiant \cite{DBLP:journals/siamcomp/Valiant79} to show the \cP-completeness of $\#\bcnf^+$. 	Recall that the main steps of Valiant's chain of reductions are:
	\begin{align*}
	\#\tcnf \leq \perm &\leq \#\perfmat\\
	&\leq \#\impmat \\
	&\leq \#\bcnf^+.
	\end{align*} 
	Our idea is to add a $\stcnf$-formula to the input of each problem in the above chain of reductions, and to express certain properties of the respective inputs in the added formula. We then count only the solutions to the input that also satisfy the added formula. 

	We first reduce $\#\stcnf$ to $\#(\tcnf, \stcnf^-)$, and then apply the above chain of reductions with the added formulae to reduce $\#(\tcnf, \stcnf^-)$ to $\#(\bcnf^-, \stcnf^-)$, which will be further reduced to $\#\scnf^-$. All of these reductions will be parsimonious, except for the one from perfect matchings to imperfect matchings.

	\underline{$\#\stcnf \leq \#(\tcnf, \stcnf^-)$}:  Let $φ(x_1, \dots, x_k) \in \stcnf$ with $k \in \mathbb{N}$ and
	\[φ(x_1, \dots, x_k) = ∃y_1 \dots ∃y_\ell \bigwedge C_i ∧ \bigwedge D_i ∧ \bigwedge E_i,\]
	where $\textrm{Var}(C_i) \subseteq \{x_1, \dots, x_k\}$, $\textrm{Var}(D_i) \subseteq \{y_1, \dots, y_\ell\}$, $\textrm{Var}(E_i) \cap \{x_1, \dots, x_k\} \!\neq\! \emptyset$ and $\textrm{Var}(E_i) \cap \{y_1, \dots, y_\ell\} \neq \emptyset$.
	We now construct two formulae $φ' \in \tcnf$ and $ψ \in \stcnf^-$ such that $\#\stcnf(φ) = \#(\tcnf, \stcnf^-)(φ', ψ)$.
	Define
	\[φ'(x_1, \dots, x_k, e_1, \dots, e_m) = \bigwedge \left( C_i ∧ \bigwedge (\neg e_i \leftrightarrow \restr{E_i}{\{x_1, \dots, x_k\}})\right) \]
  and
	\[ψ(e_1, \dots,e_m) = ∃y_1 \dots ∃y_\ell \bigwedge D_i ∧ \bigwedge (e_i \rightarrow \restr{E_i}{\{y_1, \dots, y_\ell\}}),\]
	where $m$ is the number of the clauses $E_i$, and $\restr{C}{V}$ denotes the the restriction of the clause $C$ to variables in $V$. More precisely, for a clause $C = \ell_1 \lor \ell_2 \lor \ell_3$ we define
	\[\restr{C}{V} \ \defeq \quad \quad \mathclap{\bigvee_{\substack{i \in \{1,2,3\}\\∃x \in V: \ell_i = x \textrm{ or } \ell_i = \neg x}}} \quad \ell_i.\]

  Note that in these two formulae the new implications and biimplications can be trivially transformed to \tcnf-formulae, and in ψ the free variables only occur negatively.
	Intuitively, using the new variables $e_i$, the formula $\varphi'$ expresses that the assignment to the variables $x_1, \dots, x_k$ does not satisfy any literal in $E_i$, and thus, as expressed in $\psi$, the clause $E_i$ has to be satisfied by an appropriate assignment to the variables $y_1, \dots, y_\ell$. Since the assignments to the new variables $e_i$ are uniquely determined by the assignments to the variables $x_1, \dots, x_k$, the formula $\varphi'\wedge\psi$ has the same satisfying assignments as the original formula $\varphi$.

	\problemdef{$\#(\cycc, \stcnf^-)$}{Directed Graph $G = (V, E)$ with $E = \{e_1, \dots, e_n\}$ and formula $φ(e_1, \dots, e_n) \in \stcnf^-$}{Number of cycle covers $E' \subseteq E$ of $G$ with $c_{E'} \models φ$, where $c_{E'}$ is the characteristic function of $E'$ \wrt $E$}
	
	\underline{$\#(\tcnf, \stcnf^-) \leq \#(\cycc, \stcnf^-)$}: Let $φ(x_1, \dots, x_k) \in \tcnf$, $ψ(x_1, \dots, x_k) \in \stcnf^-$. We map φ to an instance $G$ of $\#\cycc$ by a variant of the reduction from \cite{DBLP:journals/siamcomp/Valiant79}. In Valiant's reduction in \cite{DBLP:journals/siamcomp/Valiant79}, certain pairs of nodes are connected by so-called junctions, which are essentially two edges connecting the nodes in both directions. The goal then is to count only ``good'' cycle covers, namely those cycle covers containing for each junction at most one of these edges. In the original proof this is achieved by replacing junctions by a certain gadget. In our case, we can instead use a formula to express the crucial condition: A junction consisting of two edges $e_1, e_2$ is used appropriately if and only if one of the edges $e_1$ and $e_2$ is not contained in the cycle cover.
	
	In Valiant's construction, each satisfying assignment of φ corresponds to exactly one good cycle cover of $G$, and vice versa. In particular, for each variable $x$ of φ, there is a certain edge $e$ in $G$ such that $e$ is contained in each good cycle cover of $G$ if and only if the variable $x$ is assigned to 1 by the corresponding assignment.

	Now, let $ψ'$ be the formula obtained from ψ by replacing all occurrences of the free variables by the corresponding edges. Let $J$ be the set of junctions in $G$, each of which can be given as the set of its edges. Define
	\[ψ'' = ψ' ∧ \quad \mathclap{\bigwedge_{\{j_1, j_2\} \in J}}\quad (\neg j_1 ∨ \neg j_2)\]
	Note that the free variables in $ψ''$ only occur negatively. Then, we have
	\[\#(\cycc, \stcnf^-)(G,ψ'') = \#(\tcnf, \stcnf^-)(φ, ψ).\]	
	\problemdef{$\#(\perfmat, \stcnf^-)$}{Bipartite Graph $G = (V_1, V_2, E)$ with $E = \{e_1, \dots, e_n\}$ and formula $φ(e_1, \dots, e_n) \in \stcnf^-$}{Number of perfect matchings $E' \subseteq E$ of $G$ with $c_{E'} \models φ$, where $c_{E'}$ is the characteristic function of $E'$ \wrt $E$}
	
	\underline{$\#(\cycc, \stcnf^-) \leq \#(\perfmat, \stcnf^-)$}: Following the 1-to-1 correspondence between cycle covers of directed graphs and perfect matchings of bipartite graphs, the reduction can be given as follows:
	\begin{align*}
	((V, E), φ) &\mapsto
	((V, \{v' \mid v \in V\}, \{\{v_1, v_2'\} \mid (v_1, v_2) \in E\}), φ'),
	\end{align*}
	where $φ'$ is obtained from φ by replacing all occurrences of variables $(v_1, v_2)$ by the  corresponding new edges $\{v_1, v_2'\}$, which are variables in $φ'$.

	\problemdef{$\#(\impmat, \stcnf^-)$}{Bipartite Graph $G = (V_1, V_2, E)$ with $E = \{e_1, \dots, e_n\}$ and formula $φ(e_1, \dots, e_n) \in \stcnf^-$}{Number of matchings $E' \subseteq E$ of $G$ with $c_{E'} \models φ$, where $c_{E'}$ is the characteristic function of $E'$ \wrt $E$}
	
	\underline{$\#(\perfmat, \stcnf^-) \leq \#(\impmat, \stcnf^-)$}: 
	
	\noindent Let $G = (V_1, V_2, E)$ be a bipartite graph with $E = \{e_1, \dots, e_n\}$ and $ψ(e_1, \dots, e_n) \in \stcnf^-$. For the reduction $\perfmat \leq \impmat,$ 
	Valiant constructs bipartite graphs $G_k$ for $1 \leq k \leq |V_1|+1$ from $G$ by adding copies of all nodes in $V_1$ as follows:
	\[G_k = (V_{1,k}, V_2, E_k), \textrm{ where}\]
	\[V_{1, k} = V_1 \cup \{u_{ij} \mid 1 \leq i \leq |V_1|, 1 \leq j \leq k\} \textrm{ and }
	E_k = E \cup \{\{u_{ij}, v_i\} \mid 1 \leq i \leq |V_1|, 1 \leq j \leq k\})\]
	Let $A_r$ be the number of matchings of $G$ of size $|V_1|-r$. Then $G_k$ has exactly $\sum_{r=0}^{|V_1|} A_r \cdot  (k+1)^r$ matchings. Using the number of matchings of all graphs $G_k$ we get a system of linear equations that allows us to compute $A_0$, the number of perfect matchings of $G$. Note that each matching of $G$ corresponds to a number of matchings in each $G_k$ (those consisting only of copies of the edges from the original matching).
	
	To compute the number of perfect matchings $E'$ of $G$ with $c_{E'} \models ψ$, we now associate each graph $G_k$ with a formula $ψ_k$ such that $c_{E''} \models ψ_k$ holds for those matchings $E''$ of $G_k$ corresponding to a matching $E'$ of $G$ with $c_{E'} \models ψ$.
	Let $e_i = \{v_1, v_2\}$ be an edge of $G$. A matching $E''$ of $G_k$ corresponds to a matching $E'$ of $G$ that does not use edge $e_i$ if and only if it does neither use the edge $\{v_1, v_2\}$ nor any of the edges $e_{ij}$, where $e_{ij} = \{u_{1,j}, v_2\}$. Formally this can be written as
	\[c_{E'} \models \neg  e_i ⇔ c_{E''} \models \neg \{v_1, v_2\} ∧ \bigwedge_{1 \leq j \leq k} \neg e_{ij}.\]
	Now in any clause $(\neg  e_i ∨ \ell_1 ∨ \ell_2)$ where $\ell_1$ and $\ell_2$ are literals of bound variables of ψ we can replace $\neg  e_i$ by $\bigwedge_{1 \leq j \leq k} \neg  e_{ij}$. The resulting formula is equivalent to
	\[\bigwedge_{1 \leq j \leq k} (\neg  e_{ij} ∨ \ell_1 ∨ \ell_2),\]
	which is of the desired form. Similarly we can replace any clause of the form $(\neg e_{i_1} ∨ \neg  e_{i_2} ∨ \ell_1)$ by $(\bigwedge_{1 \leq j \leq k} \neg  e_{i_1,j} ∨ \bigwedge_{1 \leq j \leq k} \neg  e_{i_2, j} ∨ \ell_1)$, resulting in the formula 
	\[\bigwedge_{(j_1, j_2) \in \{1, \dots, k\}^2} (\neg  e_{i_1, j_1} ∨ \neg  e_{i_2, j_2} ∨ \ell_1).\]
	Analogously we can also handle clauses of the form $(\neg  e_{i_1} ∨ \neg  e_{i_2} ∨ \neg  e_{i_3})$.
	
	Let $ψ'$ be ψ after applying the above changes. We have that any matching $E''$ of $G_k$ corresponds to a matching of $E'$ of $G$ with $c_{E'} \models ψ$ if and only if $c_{E''} \models ψ'$. Now, we can proceed as in \cite{DBLP:journals/siamcomp/Valiant79}: Let $A_r'$ be the number of matchings $E'$ of $G$ of size $|V_1|-r$ with $c_{E'} \models ψ$. Then $G_k$ has exactly $\sum_{r=0}^{|V_1|} A_r' \cdot  (k+1)^r$ matchings $E''$ with $c_{E''} \models ψ'$. Using the number of such matchings for all graphs $G_k$ we get a system of linear equations allowing us to compute $A_0'$, the number of perfect matchings $E'$ of $G$ with $c_{E'} \models ψ$.
	
	\underline{$\#(\impmat, \stcnf^-) \leq \#(\bcnf^-, \stcnf^-)$}: Let $G = (V_1, V_2, E)$ be a bipartite graph with $E = \{e_1, \dots, e_n\}$ and $ψ(e_1, \dots, e_n) \in \stcnf^-$. The reduction works completely analogously to the proof by Valiant: We define a $\bcnf^-$-formula $φ(e_1, \dots, e_n)$ expressing that each node of the graph is only matched once as:
	\[φ(e_1, \dots, e_n) = \quad \mathclap{\bigwedge_{\substack{(e_1, e_2) \in E \times E\\e_1 \neq e_2 \textrm{ and } e_1 \cap e_2 \neq \emptyset}}} \quad \neg e_1 \lor \neg e_2\]
	Then
	\[\#(\impmat, \stcnf^-)(G, ψ) =	 \#(\bcnf^-, \stcnf^-)(φ,ψ).\]
	
	\underline{$\#(\bcnf^-, \stcnf^-) \leq \#\scnf^-$}: Let $\varphi(x_1, \dots, x_n) \in \bcnf^-$ and $\psi(x_1, \dots, x_n) \in \stcnf^-$. Furthermore, let ψ be of the form $ψ(x_1, \dots, x_n) = ∃y_1 \dots ∃y_k ψ'$. Clearly,
	\[\#(\bcnf^-, \stcnf^-)(\varphi, \psi) = \#\scnf^-(∃y_1 \dots ∃y_k (\varphi \wedge \psi')).\]
\end{proof}

\begin{proof}[Proof of Theorem \ref{thm:scnfstarNumNP}]
	Membership holds because we can simply guess an assignment and check whether it assigns some variables to $1$ and whether it satisfies the input formula.
	
	For hardness we give a reduction from $\#\scnf^-$. Let $\varphi(x_1,\dots,x_n) \in \scnf^-$. Replace every free variable $x_i$ in $\varphi$ by the constant $\bot$, and simplify the formula by removing all literals that are equivalent to $\bot$ and all clauses that are equivalent to $\top$. Denote the resulting formula by $φ'$, and
	\[ψ \defeq φ' ∧ (\neg x_{n+1} ∨ \neg x_{n+2}),\]
	where $x_{n+1}$ and $x_{n+2}$ are the only free variables. Observe that if $\varphi(x_1,\dots,x_n)$ is satisfiable then it is also satisfied by the all-0-assignment. We  use $\#\scnf_*^-$ as an oracle to compute the number of satisfying assignments of $\psi$, not counting the all-0-assignment. The answer can only be either 0 or 2. If the answer is $0$, we conclude that $\varphi(x_1,\dots,x_n)$ is not satisfiable and therefore the number of satisfying assignments of $\varphi$ is $0$. If the answer is $2$, we know that $\varphi(x_1,\dots,x_n)$ is satisfiable. Now we can ask the oracle again for the number $k$ of the satisfying assignments of $\varphi(x_1,\dots,x_n)$, and the actual number of satisfying assignments of $\varphi(x_1,\dots,x_n)$ is then  $k+1$ (as the oracle does not count the all-$0$-assignment).
\end{proof}

\end{document}